\documentclass[aps,prd,reprint,superscriptaddress,amsmath,amssymb,longbibliography,floatfix]{revtex4-2}
\usepackage{fontspec,etoolbox,amsthm,mathtools,bm,graphicx,booktabs,array,xcolor}
\usepackage{placeins}
\usepackage{orcidlink}
\usepackage{hyperref}
\definecolor{linkink}{HTML}{163C56}
\hypersetup{colorlinks=true,linkcolor=blue,citecolor=blue,urlcolor=blue}
\newtheorem{theorem}{Theorem}

\newcommand{\ii}{\mathrm i}
\newcommand{\dd}{\mathrm d}
\newcommand{\bn}{\hat{\boldsymbol n}}
\newcommand{\bu}{\hat{\boldsymbol u}}
\newcommand{\bv}{\hat{\boldsymbol v}}
\newcommand{\GY}{\widetilde Y}
\newcommand{\Ci}{\operatorname{Ci}}
\newcommand{\Si}{\operatorname{Si}}
\newcommand{\sinc}{\operatorname{sinc}}
\newcommand{\Tnull}{\mathsf T}

\makeatletter
\newenvironment{widefigure}{\begin{figure*}[!t]\centering}{\end{figure*}}
\makeatother

\makeatletter
\newenvironment{localfigure}{\par\addvspace{8pt}\noindent\begin{minipage}{\columnwidth}\centering\renewcommand{\caption}[1]{\refstepcounter{figure}\@makecaption{\fnum@figure}{##1}}}{\end{minipage}\par\addvspace{8pt}}
\newenvironment{localtable}{\par\addvspace{8pt}\noindent\begin{minipage}{\columnwidth}\centering\def\@captype{table}\renewcommand{\caption}[1]{\refstepcounter{table}\@makecaption{\fnum@table}{##1}}}{\end{minipage}\par\addvspace{8pt}}
\makeatother

\hypersetup{pdftitle={Anisotropic overlap reduction functions for arbitrary equal-arm TDI channels: A fully analytic framework and efficient evaluation}}

\newcommand{\FormulaRepositoryURL}{}
\newcommand{\FormulaRepositoryVersion}{}
\newcommand{\FormulaRepositoryDOI}{}
\newcommand{\FormulaRepositoryLink}{%
  \ifdefempty{\FormulaRepositoryURL}{\textnormal{[Formula repository URL to be added]}}{\url{\FormulaRepositoryURL}}}
\newcommand{\FormulaRepositoryCitation}{%
  \FormulaRepositoryLink; %
  \ifdefempty{\FormulaRepositoryVersion}{release version to be added}{release \FormulaRepositoryVersion}; %
  \ifdefempty{\FormulaRepositoryDOI}{archival DOI to be added}{\href{https://doi.org/\FormulaRepositoryDOI}{doi:\FormulaRepositoryDOI}}}

\begin{document}
\newcommand{\SpeedOne}{48}
\newcommand{\SpeedHigh}{4948}
\newcommand{\SweepMilliseconds}{0.311}
\newcommand{\QuadSweepMilliseconds}{360.5}
\newcommand{\SweepImprovement}{4.39}
\newcommand{\MPMidGain}{649}

\title{Anisotropic overlap reduction functions for arbitrary equal-arm TDI channels: A fully analytic framework and efficient evaluation}
\author{Kun Zhou~\orcidlink{0009-0001-1099-8623}}

\author{Jin Li~\orcidlink{0000-0001-8538-3714}}

\email{{cqujinli1983@cqu.edu.cn}}

\affiliation{College of Physics, Chongqing University, Chongqing 401331, China and Department of Physics and Chongqing Key Laboratory for Strongly Coupled Physics, Chongqing University, Chongqing 401331, China}
\begin{abstract}
We establish a unified analytic framework for anisotropic overlap reduction functions (ORFs) of arbitrary finite-delay time-delay interferometry (TDI) channels in a static, reciprocal, equal-arm equilateral constellation. The construction combines the polynomial constraints for laser-noise cancellation, constellation symmetries, and directional-projection integration. A free basis of the admissible channel module separates channel generation from angular integration; two families of elementary angular kernels then generate all channel-pair intensity ORFs. The channel algebra is independent of gravitational-wave polarization, and the intensity-kernel reduction applies to tensor, vector, breathing, and longitudinal sectors. Tensor multipoles illustrate the framework here. Complete propagation effects are retained at arbitrary fixed finite angular order, with explicit examples through $\ell=5$. Prepared frequency batches, shared multiprecision expressions, and local expansions at removable delay points provide stable evaluation. In the prepared-batch benchmark, the analytic implementation is approximately 50 to 5000 times faster than optimized sky quadrature at matched error tolerances. Independent multiprecision sky integrations also validate prescribed mixed error targets down to $10^{-100}$ in the tested cases. The resulting kernels and generation rules supply reusable inputs for anisotropic multichannel background models.
\end{abstract}
\maketitle

\section{Introduction}
The stochastic gravitational-wave background (SGWB) can arise from unresolved astrophysical sources or processes in the early Universe. Its detection seeks a statistical signal in noisy data; its characterization aims to determine the sources and mechanisms that produced it. The mean intensity spectrum describes the frequency distribution of radiation, whereas anisotropy describes its variation across the sky. For an astrophysical background, angular structure is related to the spatial distribution of sources and to large-scale structure. Different background components can also have distinct frequency-dependent angular distributions~\cite{cusin17,bartolo}. Measurements of both spectral and angular information are therefore needed to separate background components. Multiple detector channels and constellation motion provide access to this information, provided that their directional responses are modeled accurately.

Correlation analyses of the SGWB rely on a forward relation between detector responses and sky statistics~\cite{allen99,allen97,romano}. In a spherical-harmonic convention, the channel cross spectrum has the form $S_{CC'}(f)=\sum_{\ell m}p_{\ell m}(f)\Gamma_{CC'}^{\ell m}(f)$: the coefficients $p_{\ell m}$ describe the sky and the overlap reduction functions $\Gamma_{CC'}^{\ell m}$ describe how each angular component projects into a channel pair. Isotropic analyses require only the monopole kernel. Anisotropic analyses require a response matrix spanning several multipoles, channel pairs, and frequencies. This matrix enters sky reconstruction, background parameter estimation, and comparisons between readouts. Missing components or uncontrolled approximations alter the forward model. Whether its angular information suffices to distinguish different skies is a separate inference problem involving noise, observation time, and scanning modulation.

Triangular constellations are central to planned space-based observations. LISA, Taiji, and TianQin have nominal arm lengths of approximately $2.5\times10^6$, $3.0\times10^6$, and $1.7\times10^5\,\mathrm{km}$, respectively~\cite{lisa2024,taiji2026,tianqin2025}. LISA studies have considered the SGWB and its angular distribution~\cite{bartolo}. Finite propagation is controlled by $a=2\pi fL/c$. The conventional long-wavelength, low-frequency, or small-antenna approximation requires $a\ll1$, a condition that does not hold throughout every mission's broad frequency window. We consider the static, equal-arm, equilateral idealization of these constellations and retain the complete propagation response of the one-way Doppler links and delayed readouts.

Finite-arm modeling is also relevant to ground-based observatories. The Einstein Telescope (ET) includes a candidate triangular configuration with $10\,\mathrm{km}$ arms, alongside a two-site L-shaped alternative~\cite{etdesign2020,etscience2023,etgeometry2026}. For the triangle, the broad $1\,\mathrm{Hz}$--$10\,\mathrm{kHz}$ reference window corresponds to $2.10\times10^{-4}\lesssim a\lesssim2.10$, placing its upper end in the finite-arm regime~\cite{etsensitivity2023}. We include this configuration as a propagation-scale comparison; applying the angular kernels to its optical readouts requires a separate response model and is left to future work.

The required response modeling has several established foundations. Estabrook and Wahlquist derived the Doppler response of a light link; Cornish and Rubbo formulated a LISA response including finite light-travel time~\cite{estabrook75,cornish03}. TDI combines delayed one-way measurements to cancel laser frequency noise~\cite{armstrong99}. Its algebraic description expresses admissible channels as solutions of delay-polynomial constraints~\cite{dhurandhar02,tdi}. Studies of optimal channels and signal--noise separation have further demonstrated the uses of different readouts~\cite{prince02,adams10}. These developments place background analysis on physically realizable noise-canceling channels. They also motivate a practical question: how can one obtain consistent multipole ORFs when a TDI combination is changed or added, without rebuilding the full angular calculation for each new channel pair?

The anisotropic response of space interferometers has been studied in detail. Kudoh and Taruya investigated the spherical-harmonic structure of antenna power beyond the long-wavelength approximation, its geometric selection rules, and its frequency-dependent angular sensitivity. They used rotations and reflections to relate the multipole responses of standard Sagnac and AET combinations to basic auto- and cross-correlation kernels~\cite{kudoh05}. Taruya and Kudoh studied perturbative low-frequency sky-map reconstruction with a scanning constellation~\cite{taruya05}. Bartolo et al. considered LISA anisotropy measurements through multipole signal-to-noise ratios, scanning, and Fisher analyses~\cite{bartolo}. These studies establish how geometry and observing strategy constrain directional information. Low-frequency reconstruction is limited by its expansion range, while a specified channel's response depends on its readout definition, and its sensitivity also depends on the noise model. A complete analytic treatment of arbitrary finite-delay channels additionally requires their algebraic constraints to be incorporated into the angular-response calculation.

Analytic integration has also progressed. Complete frequency responses and direction-averaged expressions for space interferometers have been obtained~\cite{liang19,zhang20}, and Hu et al. developed isotropic ORF calculations retaining the full propagation effects~\cite{huGR}. Zhou and Li developed a directional-projection method for anisotropic angular integrals and supplied a public formula library~\cite{zhouli,source}. We use this real-variable integration method with the complete light-link responses and delayed channel combinations defined here.

Long-wavelength expansions and direct numerical angular integration are useful complementary approaches, each with a defined range of applicability. A low-frequency expansion simplifies the direction-dependent propagation phase. Retaining only a few leading terms, however, generally fails to describe the transition near $a\sim1$ and the higher-frequency response accurately. Conversely, short-wavelength asymptotics assume $a\gg1$ and do not replace the complete low-frequency or transition kernels. An envelope estimate also omits multipole-specific oscillations and cancellation zeros. Numerical integration retains the complete response and can be validated by grid refinement and independent implementations. Increasing phase oscillations, however, require finer angular sampling, while small responses or strong cancellations require control of absolute error. Repeated sky integration over many frequencies, angular weights, and readouts can become a computational bottleneck in parameter inference. Reusable analytic expressions are therefore desirable, although their own low-frequency cancellation must be handled carefully; an analytic form alone does not establish a universal speed advantage.

The problem addressed here is to formulate a common anisotropic ORF representation for every admissible finite-delay TDI channel in the specified geometry. This requires identifying the channel space allowed by laser-noise cancellation and stating precisely in what algebraic sense a chosen set generates it. One must then determine how constellation symmetries relate channel-pair angular integrals, in particular the relationship between general ORFs and the $11$ and $12$ elementary functions. A common representation would provide a consistent basis for comparing readouts and a reusable analytic input for anisotropic forward modeling.

Three difficulties must be resolved together. First, laser-noise cancellation is a delay-polynomial identity: linear independence at one frequency does not establish completeness for finite-delay channels. Second, cross spectra contain many link pairs and angular weights, so a direct term-by-term treatment repeats integrals and increases the risk of inconsistent phases, harmonics, or channel ordering. Third, complete endpoint differences can undergo strong low-frequency cancellation, and frequency-dependent channel transformations can have singular denominators at discrete frequencies. A reliable calculation must distinguish the physical response from apparent coordinate singularities while preserving these cancellations in both analytic evaluation and numerical tests.

Our scope is a static, reciprocal, equal-arm equilateral triangle and channels formed from finitely many commuting delays. Moving constellations, varying arm lengths, and noncommuting delays require further dynamic TDI treatment~\cite{moving04,geometric05}. The channel-module construction is polarization independent. The two-kernel intensity reduction holds for $P=T,V,B,L$; this paper develops the common framework and uses an unpolarized tensor background for its explicit illustrations. Polarization coherence requires the corresponding kernel transformation rules. The extra-polarization formulas and their physical analysis are developed in the companion work. The integration construction applies to arbitrary fixed finite $\ell,m$; the displayed range $\ell\leq5$ is an illustrative selection. Sensitivity forecasts require additional mission noise and observing models.

Section II defines the complete link response, laser-noise constraints, and multipole conventions. Section III establishes channel completeness and the relation to the $11/12$ elementary functions. Section IV presents the angular integration, a representative finite expression, and low-frequency limits. Section V examines the frequency dependence and physical interpretation of the displayed multipoles. Section VI compares Michelson, Sagnac, and Relay readouts, discusses general delayed combinations and AET, and evaluates accuracy and computational efficiency. Section VII discusses applications and extensions.

The accompanying formula library~\cite{TDIFormulaLibrary} contains the final explicit expressions and supporting numerical data used in this paper; its public address is reserved as \FormulaRepositoryLink. A Jupyter notebook provides a component index and direct display of the final formulas. The catalogue order and the validation coverage of each polarization sector are stated separately.

\section{Complete response and angular conventions}
\subsection{Link response and laser noise}
Choose the vertices
\begin{equation}
 \begin{split}
 \boldsymbol x_1&=L(0,0,0),\qquad \boldsymbol x_2=L(1,0,0),\\
 \boldsymbol x_3&=L(1/2,\sqrt3/2,0),
 \end{split}
\end{equation}
with the $z$ axis normal to the constellation plane. Adopt the plane-wave convention $e^{+\ii\omega(t-\bn\cdot\boldsymbol x/c)}$, and write $a=\omega L/c$ and $z=e^{-\ii a}$. The measurement $y_{ij}$ is the fractional frequency shift measured at $j$ for light emitted at $i$. The arm direction $\bu_{ij}=(\boldsymbol x_i-\boldsymbol x_j)/L$ points opposite to light propagation. For the standard tensor polarizations $A=+,\times$, normalized by $e^A:e^{A'}=2\delta^{AA'}$, the full response is
\begin{align}\label{eq:link}
 F_{ij}^A&=-\frac12\bu_{ij}\bu_{ij}:e^A\,
 q_a(1+\bn\cdot\bu_{ij})
 e^{-\ii a\bn\cdot\boldsymbol x_j/L},\\
 q_a(d)&=\frac{1-e^{-\ii ad}}d
       =\ii a e^{-\ii ad/2}\sinc(ad/2).
       \label{eq:q}
\end{align}
We define $q_a(0)=\ii a$. The endpoint difference is retained, so the full link response remains finite when its projection denominator vanishes.

An arbitrary finite-delay channel has the form
\begin{equation}\label{eq:channel}
 C=\boldsymbol d_C\boldsymbol y,\qquad
 \boldsymbol d_C(z)\in\mathbb C[z]^6,
\end{equation}
with links ordered as $(12,21,13,31,23,32)$. The laser contribution is $y_{ij}^{\rm laser}=z c_i^{\rm laser}-c_j^{\rm laser}$. The TDI condition is therefore
\begin{equation}\label{eq:noise}
 \boldsymbol d_C B=0,\qquad
 B=\begin{pmatrix}
 z&-1&0\\-1&z&0\\z&0&-1\\-1&0&z\\
 0&z&-1\\0&-1&z
 \end{pmatrix}.
\end{equation}
Equation~\eqref{eq:noise} must hold as a polynomial identity; a combination that cancels noise only at isolated frequencies does not satisfy this definition. The three equal-arm round-trip difference channels are
\begin{align}\label{eq:M}
 M_1&=y_{21}+zy_{12}-y_{31}-zy_{13},\nonumber\\
 M_2&=y_{32}+zy_{23}-y_{12}-zy_{21},\\
 M_3&=y_{13}+zy_{31}-y_{23}-zy_{32}.\nonumber
\end{align}
Each satisfies Eq.~\eqref{eq:noise}. If a convention instead uses $X_i=(1-z^2)M_i$, its correlation functions must include the associated delay-filter factors.

\subsection{Multipole response and physical cross spectra}
With the Condon--Shortley phase and $\GY_{\ell m}=\sqrt{4\pi}Y_{\ell m}$, define
\begin{align}\label{eq:ORF}
 \Gamma_{CC'}^{T,\ell m}
 &=\frac1{8\pi}\int\GY_{\ell m}(\bn)
       \sum_{A=+,\times}F_C^A F_{C'}^{A*}\dd\Omega,\\
 F_C^A&=\sum_\lambda d_{C\lambda}F_\lambda^A.\nonumber
\end{align}
The harmonic inside the integral is not conjugated. If the directional power of each tensor polarization is $\mathcal P_T=\sum_{\ell m}p_{\ell m}\GY_{\ell m}$, using the matching statistical prefactor gives $S_{CC'}=\sum_{\ell m}p_{\ell m}\Gamma_{CC'}^{T,\ell m}$. A real intensity distribution obeys $p_{\ell,-m}=(-1)^mp_{\ell m}^*$. A complete nonnegative intensity distribution yields a positive semidefinite channel covariance; an individual multipole matrix need not be positive semidefinite.

An equivalent real angular basis is $S_{\ell m}^{c,s}=\sqrt2(\Re,\Im)\GY_{\ell m}$ for $m>0$. Its responses are
\begin{align}\label{eq:realharmonic}
 \Gamma_{\ell m}^{c}
 &=\frac{\Gamma_{\ell m}+(-1)^m\Gamma_{\ell,-m}}{\sqrt2},\\
 \Gamma_{\ell m}^{s}
 &=\frac{\Gamma_{\ell m}-(-1)^m\Gamma_{\ell,-m}}{\ii\sqrt2}.
\end{align}
For $m=0$, we use $S^c_{\ell0}=\GY_{\ell0}$. Figure~\ref{fig:sky} illustrates positive intensity templates. Here $\ell$ specifies angular order and $m$ labels an azimuthal component in the chosen frame. The constant term keeps the intensity nonnegative, while the harmonic selects the angular structure to be studied. The corresponding multipole ORFs project this structure into the channel spectra. Signed multipole responses therefore do not represent positive or negative physical intensities.

\begin{widefigure}
\includegraphics[width=.82\textwidth,height=.27\textheight,keepaspectratio]{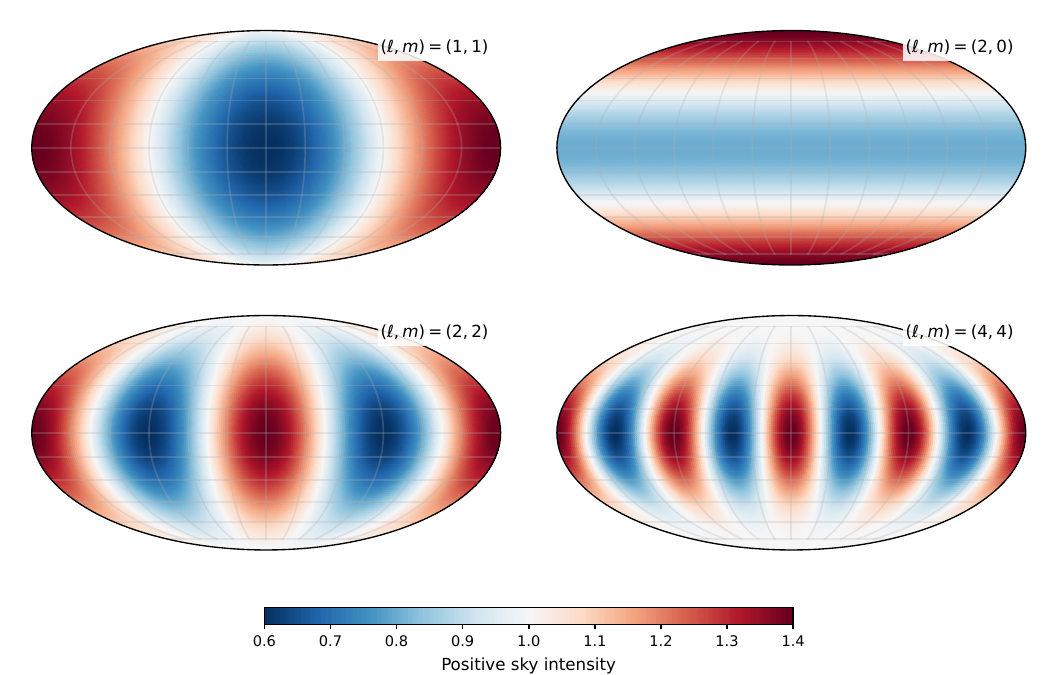}
\caption{Positive intensity templates $1+0.4S^c_{\ell m}/\max|S^c_{\ell m}|$. The $(1,1)$ template is an in-plane dipole, $(2,0)$ is a latitudinal quadrupole, and $(2,2)$ and $(4,4)$ show different azimuthal structures. Each template contributes to channel cross spectra through its multipole ORF.}
\label{fig:sky}
\end{widefigure}

\section{Generating arbitrary TDI ORFs from the $11/12$ functions}

\subsection{Two elementary functions for arbitrary equal-arm TDI}
\label{sec:gen-theorem}

Consider any finite delay-polynomial channel on a static, reciprocal, equal-arm equilateral triangle that cancels all three laser noises as a polynomial identity. Write $z=e^{-\ii a}$, order the links as $(12,21,13,31,23,32)$, and define
\begin{equation}\label{eq:gen-sagnac}
 S=y_{12}-y_{21}-y_{13}+y_{31}+y_{23}-y_{32}.
\end{equation}
This is the equal-arm symmetric Sagnac channel with its common delay removed. It obeys
\begin{equation}\label{eq:gen-sum}
 M_1+M_2+M_3=(z-1)S.
\end{equation}
Channel completeness and threefold rotational covariance together determine the basic angular integrals required for any readout.

The construction has three distinct layers. Laser-noise cancellation fixes the channel module, geometric transformations reduce its intensity kernels, and polarization contractions determine the functions that are integrated. The first layer is independent of polarization; the displayed matrix applies to the nonchiral intensity sectors specified in the theorem. General polarization coherence retains the linear channel transformation but requires its own angular transformation rules.

\begin{theorem}\label{thm:gen}
For any $\boldsymbol d=(p,q,r,s,t,u)\in\mathbb C[z]^6$ satisfying $\boldsymbol d B=0$ identically, there are unique polynomials $A,B_C,C$ such that
\begin{align}\label{eq:gen-polybasis}
 C_{\boldsymbol d}&=A M_1+B_C M_2+C S,\\
 A&=q+t,\quad B_C=-p-r,\quad C=q+s+t.
 \label{eq:gen-polyweights}
\end{align}
For $z\ne1$ the same channel has the unique representation $C_{\boldsymbol d}=\sum_i c_iM_i$, with
\begin{equation}\label{eq:gen-weights}
 c_1=A+\frac{C}{z-1},\quad
 c_2=B_C+\frac{C}{z-1},\quad
 c_3=\frac{C}{z-1}.
\end{equation}
For an equilateral triangle, the ORFs of any pair of such channels, at arbitrary angular order, are therefore generated by two elementary functions, $g_{\ell m}^P=\Gamma_{11}^{P,\ell m}$ and $h_{\ell m}^P=\Gamma_{12}^{P,\ell m}$:
\begin{equation}\label{eq:gen-final}
 \Gamma_{CC'}^{P,\ell m}
 =\boldsymbol c_C\,\mathsf K_{\ell m}^P\,
 \boldsymbol c_{C'}^\dagger.
\end{equation}
Here $P=T,V,B,L$ labels the polarization intensity kernels defined in this paper and its companion, $\epsilon=(-1)^\ell$, $\varphi=e^{2\pi\ii m/3}$, and
\begin{equation}\label{eq:gen-kmatrix}
 \mathsf K_{\ell m}^P=
 \begin{pmatrix}
 g&h&\epsilon\varphi^2h\\
 \epsilon h&\varphi g&\varphi h\\
 \varphi^2h&\epsilon\varphi h&\varphi^2g
 \end{pmatrix}.
\end{equation}
At the discrete frequencies $z=1$, Eq.~\eqref{eq:gen-final} is understood as the continuous limit of the full response of the finite-delay channel.
\end{theorem}

\textit{Proof of channel completeness.} The three noise constraints are
\begin{align}\label{eq:gen-noise-components}
 z(p+r)-q-s&=0,\\
 -p+z(q+t)-u&=0,\\
 -r+zs-t+zu&=0.
\end{align}
Their sum is $(z-1)(p+q+r+s+t+u)=0$. Since the polynomial ring has no zero divisors, $p+q+r+s+t+u=0$. The six link coefficients of $A M_1+B_CM_2+C S$ are
\begin{equation}\label{eq:gen-expanded}
 \begin{split}
 (&zA-B_C+C,\ A-zB_C-C,\ -zA-C,\\
  &-A+C,\ zB_C+C,\ B_C-C).
 \end{split}
\end{equation}
Substitution of Eq.~\eqref{eq:gen-polyweights}, followed by use of the noise constraints and vanishing coefficient sum, gives $(p,q,r,s,t,u)$. Conversely, applying the three coefficient extractions $q+t$, $-p-r$, and $q+s+t$ to Eq.~\eqref{eq:gen-expanded} returns $A,B_C,C$, respectively. The representation is unique. Equivalently, the matrix $U$ whose rows are the link coefficients of $M_1,M_2,S$ has a $z$-independent right inverse $R$, with $UR=I_3$. These channels thus form a free basis of the polynomial module of laser-noise-canceling channels. The argument also applies to Laurent polynomials if integer advances are allowed.

Equation~\eqref{eq:gen-sum} immediately gives the rational weights in Eq.~\eqref{eq:gen-weights}. The three $M_i$ form a basis over $\mathbb C(z)$ but do not generate the full polynomial module. In particular, the unique $M$ weights of $S$ are all $(z-1)^{-1}$. It would therefore be incorrect to identify all finite TDI channels with polynomial combinations of the three $M_i$. The ordinary vertex Sagnac channel is another useful example:
\begin{align}\label{eq:gen-alpha}
 \alpha_1={}&y_{21}+z y_{32}+z^2y_{13}
 -y_{31}-z y_{23}-z^2y_{12}\\
 ={}&(1-z)M_1-zS.
\end{align}
Its two three-link paths have the same laser residual $(z^3-1)c_1^{\rm laser}$, which cancels in their difference.

\textit{Proof of the two-function matrix.} A rotation through $2\pi/3$ about the triangle center sends $1\to2\to3\to1$ and $M_i\to M_{i+1}$. The polarization intensity kernel transforms covariantly; the common response phase from a translation of the coordinate origin cancels in the product. Changing the integration variable and using $\GY_{\ell m}(R\bn)=\varphi\GY_{\ell m}(\bn)$ gives
\begin{equation}\label{eq:gen-rotation}
 \Gamma_{i+1,j+1}^{P,\ell m}
 =\varphi\Gamma_{ij}^{P,\ell m},
\end{equation}
with indices understood modulo 3.

An exchange relation follows from the complete propagation response. For the real polarization tensors used here, choose $e^A(-\bn)=\eta_Ae^A(\bn)$, with $\eta_A=\pm1$. The full single-link response satisfies
\begin{equation}\label{eq:gen-link-inversion}
 F_{ij}^A(-\bn)=-\eta_Az F_{ji}^{A*}(\bn).
\end{equation}
This follows directly from $q_a(d)=-e^{-\ii ad}q_a^*(d)$ and $\boldsymbol x_i-\boldsymbol x_j=L\bu_{ij}$, without splitting the endpoint difference. The equal-arm round-trip channels consequently obey
\begin{equation}\label{eq:gen-m-inversion}
 F_{M_i}^A(-\bn)=-\eta_Az^2F_{M_i}^{A*}(\bn).
\end{equation}
The summed kernels satisfy $K_{ij}(-\bn)=K_{ji}(\bn)$. With $\GY_{\ell m}(-\bn)=(-1)^\ell\GY_{\ell m}(\bn)$, this yields
\begin{equation}\label{eq:gen-exchange}
 \Gamma_{ji}^{P,\ell m}=(-1)^\ell\Gamma_{ij}^{P,\ell m}.
\end{equation}
Equations~\eqref{eq:gen-rotation} and \eqref{eq:gen-exchange} generate all nine entries of Eq.~\eqref{eq:gen-kmatrix}; for odd $\ell$ they also force $g=0$. Reflection in the detector plane eliminates intensity multipoles with odd $\ell+m$. The general negative-$m$ relation is $\Gamma_{ij}^{\ell,-m}=(-1)^m\Gamma_{ji}^{\ell m*}$, which reduces to $\Gamma_{ij}^{\ell,-m}=\Gamma_{ij}^{\ell m*}$ for the allowed nonzero components. Finally, substituting the linear channel expansion into the angular integral proves Eq.~\eqref{eq:gen-final}.

\subsection{Mirror constraints and independent elementary functions}
\label{sec:symmetry-independent}
An in-plane reflection supplies a further constraint beyond rotation and exchange.
In the adopted coordinates the affine position map
$\mathcal M(\boldsymbol x)=\mathsf R\boldsymbol x+L(1,0,0)$,
with $\mathsf R=\operatorname{diag}(-1,1,1)$, exchanges vertices 1 and 2 and fixes vertex 3.
Directions transform only under its linear part, $\bn\mapsto\mathsf R\bn$,
or $\phi\mapsto\pi-\phi$. Write $\sigma=(12)$.
Reflection interchanges the two arms of each round-trip difference channel;
the resulting common minus sign and the common translation phase cancel in the intensity kernel.
Substitution in the complete link response gives
\begin{equation}\label{eq:sym-mirror-response}
 F_{M_i}^{A}(\mathsf R\bn)=-\rho_A e^{\ii a n_x}F_{M_{\sigma i}}^{A}(\bn),\quad\rho_A=\pm1.
\end{equation}
The polarization reflection sign cancels in each same-mode pairing. For the nonchiral polarization intensity sectors considered here,
\begin{equation}\label{eq:sym-mirror-kernel}
 K_{ij}^{P}(\mathsf R\bn)=K_{\sigma i,\sigma j}^{P}(\bn).
\end{equation}
Since $\GY_{\ell m}(\mathsf R\bn)=\GY_{\ell,-m}(\bn)$ in the Condon--Shortley convention,
\begin{equation}\label{eq:sym-mirror}
 \Gamma_{ij}^{P,\ell m}=\Gamma_{\sigma i,\sigma j}^{P,\ell,-m}.
\end{equation}
The general pairing identity and the inversion property of the elementary round-trip kernels give
\begin{align}\label{eq:sym-conjugate}
 \Gamma_{ij}^{P,\ell,-m}
 &=(-1)^m\Gamma_{ji}^{P,\ell m*},\nonumber\\
 \Gamma_{ij}^{P,\ell m*}
 &=(-1)^{\ell+m}\Gamma_{ij}^{P,\ell,-m}.
\end{align}
These relations are consistent with the established geometric selection rules~\cite{bartolo,kudoh05}.
The second identity requires the inversion property of the complete round-trip response;
field reality alone does not extend it to arbitrary delayed combinations.

Combining reflection, exchange, and $C_3$ covariance yields, for the allowed nonzero components,
\begin{align}\label{eq:sym-lock}
 h_{\ell m}^{P*}&=(-1)^\ell h_{\ell m}^{P},\nonumber\\
 g_{\ell m}^{P*}&=e^{2\pi\ii m/3}g_{\ell m}^{P},\qquad \ell\ {\rm even}.
\end{align}
Thus $h$ is real for even $\ell$ and purely imaginary for odd $\ell$, whereas
$e^{\ii m\pi/3}g$ is real. The equivalent amplitude $e^{-\ii m\pi/6}g$ is also real:
the allowed $g$ components have even $m$, and the two real conventions differ by $(-1)^{m/2}$.
Each allowed $11$ component therefore contains one real amplitude and a fixed azimuthal phase.

For each polarization intensity sector, an independent set is
\begin{align}\label{eq:sym-independent}
 g:&\quad \ell=2k,\quad m=0,2,\ldots,2k,\nonumber\\
 h:&\quad 0\leq m\leq\ell,\quad\ell+m\ {\rm even}.
\end{align}
The negative-$m$ components follow from Eq.~\eqref{eq:sym-conjugate}, and the two-function matrix generates the remaining channel pairs.
For a displayed upper order $\ell_{\max}$, with $k_{\max}=\lfloor\ell_{\max}/2\rfloor$,
the numbers of real amplitudes are
\begin{equation}\label{eq:sym-count}
 N_g=\frac{(k_{\max}+1)(k_{\max}+2)}2,
 \quad N_h=\left\lfloor\frac{(\ell_{\max}+2)^2}{4}\right\rfloor.
\end{equation}
At $\ell_{\max}=5$ this gives $6+12=18$ elementary functions, each with one real degree of freedom.
This reduction concerns the response representation and its validation; it does not count sky parameters recoverable from noisy data.
The zeros and phase constraints refer to the adopted detector frame. A general rigid rotation mixes $m$ through a Wigner transformation.
Channels with arbitrary complex delay weights are assembled using the generation formula and need not preserve all selection rules of the elementary autocorrelations.

\subsection{Discrete frequencies and scope}
\label{sec:gen-limits}
The greatest common divisor of the $3\times3$ minors of the coefficient matrix $W$ of the three $M_i$ is $z-1$. For example, the minor formed by columns $(12,13,23)$ is $-z(z-1)$. On the physical unit circle, $W$ loses rank only at $z=1$, where its rank is 2; it retains rank 3 at $z=-1$. Solving from a selected pair of link coefficients can introduce a denominator $1-z^2$ and thus an additional coordinate singularity at $z=-1$, but this is not a required basis degeneracy. The matrix $U=(M_1,M_2,S)$ has rank 3 for every $z$. At $z=1$ the single-frequency noise matrix $B$ itself loses rank, and its null space becomes four dimensional. This does not raise the rank of the finite TDI polynomial module: channels satisfying the noise condition identically remain generated by $U$.

It is useful to express a channel in the polynomial free basis $U$ first, so that coordinate degeneracies are distinguished from the physical response. Define
\begin{equation}\label{eq:gen-H}
 H=\begin{pmatrix}1&0&0\\0&1&0\\
 (z-1)^{-1}&(z-1)^{-1}&(z-1)^{-1}\end{pmatrix}.
\end{equation}
Then $\Gamma_U=H\mathsf K H^\dagger$ has removable singularities at these discrete frequencies and is defined there by continuity. Two short examples are
\begin{align}\label{eq:gen-S-elements}
 \Gamma_{M_1S}^{P,\ell m}
 &=\frac{g+(1+\epsilon\varphi^2)h}{z^*-1},\\
 \Gamma_{SS}^{P,\ell m}
 &=\frac{(1+\varphi+\varphi^2)[g+(1+\epsilon)h]}
 {|z-1|^2}.
\end{align}
The second expression permits only $m$ divisible by 3 in the Sagnac autocorrelation ORF, and vanishes for odd $\ell$. At $a=2\pi n$, numerator and denominator must be taken to their common limit rather than evaluated separately at zero. The complete link function $q_a$ is entire in $a$ and obeys $|q_a(d)|\leq|a|$ for real $a,d$. The finite delay coefficients are bounded on compact frequency intervals, so the continuous limit can be taken under the sky integral. This independently guarantees finite physical ORFs. Although the free basis itself is nondegenerate, its rational ORF representation still requires the cancellations to be retained, either through continuous limits or through local series.

Let $N(a)$ and $Q(a)$ denote the numerators of the first and second expressions in Eq.~\eqref{eq:gen-S-elements}. The complete responses imply $N(a_n)=0$ and $Q(a_n)=Q'(a_n)=0$. Since $z^*-1=\ii(a-a_n)+O((a-a_n)^2)$, the removable values are
\begin{equation}\label{eq:gen-removable}
 \begin{split}
 \Gamma_{M_1S}(a_n)&=-\ii N'(a_n),\\
 \Gamma_{SS}(a_n)&=\tfrac12Q''(a_n),\qquad a_n=2\pi n .
 \end{split}
\end{equation}

The two-function construction relies on equilateral geometry and a single reciprocal delay $z$. Unequal arms generally break the threefold rotational symmetry and alter the channel module, requiring a new channel basis and elementary angular kernels. Noncommuting delays in dynamic TDI lie outside the polynomial theorem. Within the stated assumptions, there is no upper bound on the number of channels, finite delay-polynomial degree, or finite angular order. A rigid rotation into another inertial frame requires a standard Wigner rotation within each $\ell$, rather than new angular integrations. The displayed and tabulated range $\ell\leq5$ is a choice of examples, not a limitation of the theorem or integration method.

\section{Analytic elementary functions}

\subsection{Directional projections and analytic integration}
\label{sec:public-integral}
We evaluate the angular integrals using the directional-projection method of Zhou and Li~\cite{zhouli}. For two nonparallel arm directions $\bu$ and $\bv$, define $c_{uv}=\bu\cdot\bv$, $s_{uv}=\sqrt{1-c_{uv}^2}>0$, and the projections $\xi=\bn\cdot\bu$ and $\zeta=\bn\cdot\bv$. Their allowed values form the ellipse
\begin{equation}
 E=\{(\xi,\zeta):\Delta\geq0\},\qquad
 \Delta=s_{uv}^2-\xi^2-\zeta^2+2c_{uv}\xi\zeta .
\end{equation}
With $\boldsymbol e_2=(\bv-c_{uv}\bu)/s_{uv}$ and $\boldsymbol e_3=\bu\times\bv/s_{uv}$, the two points on the sphere with a given pair of projections are
\begin{equation}
 \bn_\pm=\xi\bu+
 \frac{\zeta-c_{uv}\xi}{s_{uv}}\boldsymbol e_2
 \pm\frac{\sqrt\Delta}{s_{uv}}\boldsymbol e_3,\qquad
 \dd\Omega_\pm=\frac{\dd\xi\dd\zeta}{\sqrt\Delta}.
\end{equation}
The basic link-pair integral consequently takes the form
\begin{align}\label{eq:public-ellipse}
 G_{\lambda\mu}^{P,\ell m}
 ={}&\frac1{8\pi}\int_E\frac{\dd\xi\dd\zeta}{\sqrt\Delta}
 C_P(\xi,\zeta)q_a(1+\xi)
 q_a(1+\zeta)^* \nonumber\\
 &\times\sum_{\epsilon=\pm}\GY_{\ell m}(\bn_\epsilon)
 e^{-\ii(\omega/c)\bn_\epsilon\cdot
 (\boldsymbol r_\lambda-\boldsymbol r_\mu)}.
\end{align}
Here $\boldsymbol r_\lambda$ is the receiver position and $C_P$ is the contraction numerator after the specified polarization sum. Both branches of the sphere and the receiver-position phase must be retained. Parallel and antiparallel arm pairs are obtained from the continuous limits of the complete responses. Each projection integral is evaluated in a frame adapted to the arm pair and rotated back to the common frame using spherical harmonics.

There is no angular-order cutoff in this construction. For each fixed finite $\ell$, a spherical harmonic is the restriction of a homogeneous harmonic polynomial of degree $\ell$. In the coplanar geometry, the receiver-position phase is identical on the two branches. Adding those branches cancels odd powers of the normal component, while even powers become integer powers of $\Delta$. Increasing $\ell$ therefore increases the degree and number of projection moments entering the same real angular integrals; it does not require a different angular-integration principle. The channel-generation theorem likewise holds for arbitrary $\ell,m$. The explicit formula library and plots accompanying this paper stop at $\ell=5$ as an illustrative implementation; formulas at higher order have not all been tabulated in that release.

The tabulated equilateral elementary functions are finite combinations of exponentials, frequency powers, $\ln a$, $\Ci$, $\Si$, and $E_1$. The final expressions use the geometric, polarization, and delay conventions specified here and are checked independently by complete light-link sky integration. Lengthy expressions are supplied in an electronic formula library; the main text gives representative finite expressions and physical limits.

\subsection{Tensor numerator and a finite expression}
After summing over tensor polarizations, the directional numerator is
\begin{equation}\label{eq:tensor-numerator}
 C_T(\xi,\zeta)=\frac14
 \left[2(c_{uv}-\xi\zeta)^2-(1-\xi^2)(1-\zeta^2)\right].
\end{equation}
This numerator and the complete propagation functions determine the elementary angular integrals; the TDI generation coefficients do not enter that integration step. A representative finite expression is the self-correlation monopole. For $a>0$, define
\begin{align}
 \mathcal C_1(a)&=\Ci(a)-\Ci(2a)+\ln2,\nonumber\\
 \mathcal C_2(a)&=-\Ci(a)+2\Ci(2a)-\Ci(3a)
                  +\ln(3/4),\label{eq:Cfunctions}\\
 \mathcal S_2(a)&=\Si(a)-2\Si(2a)+\Si(3a).\nonumber
\end{align}
Then
\begin{align}\label{eq:tensor-closed}
 g_{00}^{T}={}&\frac32\left[
 \mathcal C_2\cos(2a)-\mathcal S_2\sin(2a)\right]
 +\mathcal C_1 \nonumber\\
 &+\frac{5+\cos(2a)}{12}
 +\frac{2\sin(2a)-7\sin a}{4a}\nonumber\\
 &+\frac{5\cos a-4\cos(2a)-4}{4a^2}
 +\frac{4\sin(2a)-5\sin a}{4a^3}.
\end{align}
The inverse-frequency terms cancel at zero frequency, where the continuous value is zero. Direct evaluation in machine precision can lose this cancellation; the low-frequency series below or arbitrary-precision arithmetic should therefore be used in this regime. The remaining tabulated multipoles have analogous finite special-function expressions supplied individually in the electronic library. Higher orders use the same integration procedure with the corresponding directional polynomials and moments.

\subsection{Low-frequency coefficients and selection rules}
The monopoles expand as
\begin{align}\label{eq:tensor-low}
 g_{00}^{T}&=\frac35a^2-\frac{169}{840}a^4+
             \frac{85}{3024}a^6+O(a^8),\\
 h_{00}^{T}&=-\frac3{10}a^2+\frac{169}{1680}a^4-
             \frac{85}{6048}a^6+O(a^8).\nonumber
\end{align}
At low frequency the round-trip difference is described by an in-plane traceless response tensor $\mathcal D$. A two-dimensional traceless matrix obeys $\mathcal D^2\propto$ the plane projector. The leading self-correlation angular kernel therefore contains only $\ell=0,2,4$ and $m=0,\pm4$. For example,
\begin{align}\label{eq:soft-multipoles}
 \lim_{a\to0}\frac{g_{20}^{T}}{a^2}
 &=\frac{6\sqrt5}{35},&
 \lim_{a\to0}\frac{g_{40}^{T}}{a^2}
 &=\frac1{70},\\
 \lim_{a\to0}\frac{g_{22}^{T}}{a^2}
 &=0,&
 \lim_{a\to0}\frac{g_{42}^{T}}{a^2}
 &=0.\nonumber
\end{align}
Finite arm length generates the last two components at higher order. Odd cross-correlation multipoles also have different suppression orders:
\begin{align}\label{eq:odd-low}
 h_{11}^{T}&=-\frac{\ii\sqrt6}{336}a^5+O(a^7),\nonumber\\
 h_{33}^{T}&=\frac{\ii\sqrt{35}}{120}a^3+O(a^5),\\
 h_{53}^{T}&=\frac{\ii\sqrt{385}}{9240}a^3+O(a^5).\nonumber
\end{align}
For this channel pair, the octupole kernel can therefore exceed the dipole kernel in the long-wavelength regime. A lower multipole order does not by itself imply a larger response: the amplitude also depends on geometric selection rules and the first nonzero order in the propagation expansion.

\section{Frequency dependence and physical interpretation}

\subsection{Physical scales and the plotted frequency range}
The dimensionless light-travel parameter is related to physical frequency by
\begin{equation}\label{eq:afrequency}
 a=\frac{2\pi fL}{c}=\frac{f}{f_*},\qquad
 f_*=\frac{c}{2\pi L}.
\end{equation}
Table~\ref{tab:instruments} converts representative physical arm lengths and broad reference frequency windows into $a$. Figure~\ref{fig:instrument-ranges} compares them with the interval used for the ORF plots. The $10^{-4}$--$1\,\mathrm{Hz}$ windows for LISA, Taiji, and TianQin follow their design literature~\cite{lisa2024,taiji2026,tianqin2025}. For TianQin we use the nominal nearly equilateral geometry, $L\simeq\sqrt3\times10^5\,\mathrm{km}$. The ET triangle is included only as a ground-based propagation-scale reference: $L=10\,\mathrm{km}$ and the $1$--$10^4\,\mathrm{Hz}$ window used in the cited sensitivity data~\cite{etscience2023,etsensitivity2023}. These windows do not represent uniform or noise-weighted SGWB sensitivity. The length in $a$ is the physical arm length; optical storage is treated through the appropriate readout transfer function.

\begin{localtable}
\caption{Representative detector propagation scales. ET denotes the candidate triangle and is included for scale comparison only. Frequency intervals are broad reference windows.}
\label{tab:instruments}
\begin{ruledtabular}
\begin{tabular}{lccc}
Detector & $L$ ($\mathrm{km}$) & $f$ ($\mathrm{Hz}$) & Range of $a$\\
LISA & $2.5\times10^6$ & $10^{-4}$--$1$ & $0.00524$--$52.4$\\
Taiji & $3.0\times10^6$ & $10^{-4}$--$1$ & $0.00629$--$62.9$\\
TianQin & $\simeq1.73\times10^5$ & $10^{-4}$--$1$ & $0.000363$--$3.63$\\
ET triangle & $10$ & $1$--$10^4$ & $0.000210$--$2.10$
\end{tabular}
\end{ruledtabular}
\end{localtable}

\begin{localfigure}
\includegraphics[width=\columnwidth]{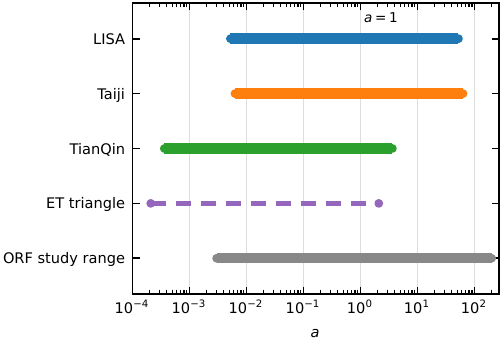}
\caption{Propagation-scale intervals for three space constellations and the $10\,\mathrm{km}$ ET triangle. The purple dashed segment is the ET reference; it does not denote a computed ET optical response. The gray bar gives the ORF sampling interval, $0.003\leq a\leq200$, and the vertical dashed line marks $a=1$. These intervals are not noise-weighted sensitivity bands.}
\label{fig:instrument-ranges}
\end{localfigure}

The space-mission windows include the finite-arm transition near $a\sim1$, so a leading long-wavelength approximation cannot describe them uniformly. For the ET triangle, $a=1$ corresponds to $f\simeq4.77\,\mathrm{kHz}$: its lower-frequency region has $a\ll1$, while its upper reference band reaches $a=O(1)$.

The main ORF curves span $0.003\leq a\leq200$ to resolve the transition and higher-frequency structure and to test analytic evaluation. The upper endpoint corresponds to approximately $3.82$, $3.18$, and $55.1\,\mathrm{Hz}$ for LISA, Taiji, and TianQin, respectively, above their listed windows. Lower-frequency TDI responses can be evaluated with the low-frequency expansions or stable complete expressions.

\subsection{Even multipoles and reference-frame phases}
Figures~\ref{fig:self}--\ref{fig:odd} display all 18 independent nonnegative-$m$ components of the $11/12$ functions within the chosen example range $\ell\leq5$. We plot $\Gamma/a^2$ over $0.003\leq a\leq200$, with spacing no larger than $0.15$ for $a\geq3$. Open symbols are complete light-link sky integrations. Reflection in the detector plane forbids odd $\ell+m$, and the inversion relation for the round-trip channels forbids odd-$\ell$ self responses. These statements concern the elementary channels; arbitrary delayed combinations are obtained from Eq.~\eqref{eq:gen-final}.

The six self-correlation components probe the mean intensity, latitudinal weighting, and different azimuthal patterns. The $g_{22}^{T}$ and $g_{42}^{T}$ components are suppressed in the leading long-wavelength limit but become nonzero at finite frequency. Propagation phases modify the directional weight seen by the detector. The self-correlation angular kernel is real, yet its projection onto complex harmonics can carry a fixed complex phase. In our frame,
\begin{equation}
 e^{-\ii m\pi/6}g_{\ell m}^{T}\in\mathbb R.
\end{equation}
This phase describes the azimuthal orientation of a harmonic relative to the detector, rather than an imaginary part of a physical autospectrum. A physical auto spectrum follows from the real combinations in Eq.~\eqref{eq:realharmonic}, or from a consistent sum over positive and negative $m$.

\begin{widefigure}
\includegraphics[width=\textwidth]{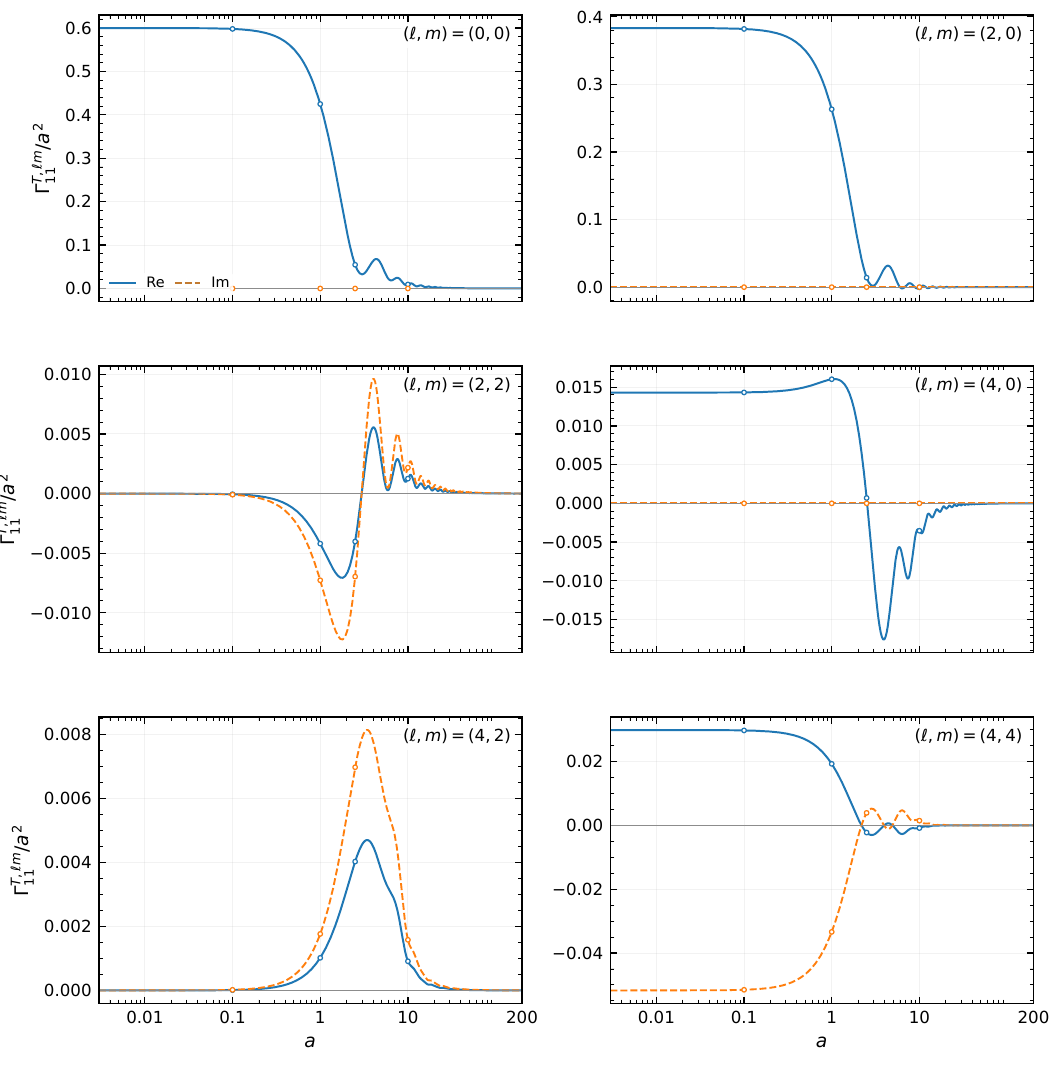}
\caption{All independent nonnegative-$m$ tensor $11$ self-correlation components in the example range $\ell\leq5$. Solid and dashed curves show real and imaginary parts; open symbols are independent sky integrations. Signed multipole projections are retained. The low-frequency suppression of $(2,2)$ and $(4,2)$ follows Eq.~\eqref{eq:soft-multipoles}.}
\label{fig:self}
\end{widefigure}

All six even cross-correlation components are real, as shown in Fig.~\ref{fig:even}. Their signs describe in-phase or antiphase channel projections for the specified angular weight. Zeros and oscillations arise jointly from finite propagation phases and the readout combination. Different multipoles have different frequency profiles, so frequency information and directional information are coupled. Their contribution to a sky reconstruction still depends on the noise covariance, observing time, and correlations among the sky templates.

The displayed normalization removes the common $a^2$ behavior of a leading Doppler correlation. It makes the onset of higher-order components easy to identify but does not impose a normalization on a physical background or its inferred angular power. In particular, the zeros of one multipole do not imply that a channel is blind to a general sky at that frequency: the full cross spectrum is a sum over all sky components. Likewise, a sign change of a projected kernel does not require a negative sky intensity.

\begin{widefigure}
\includegraphics[width=\textwidth]{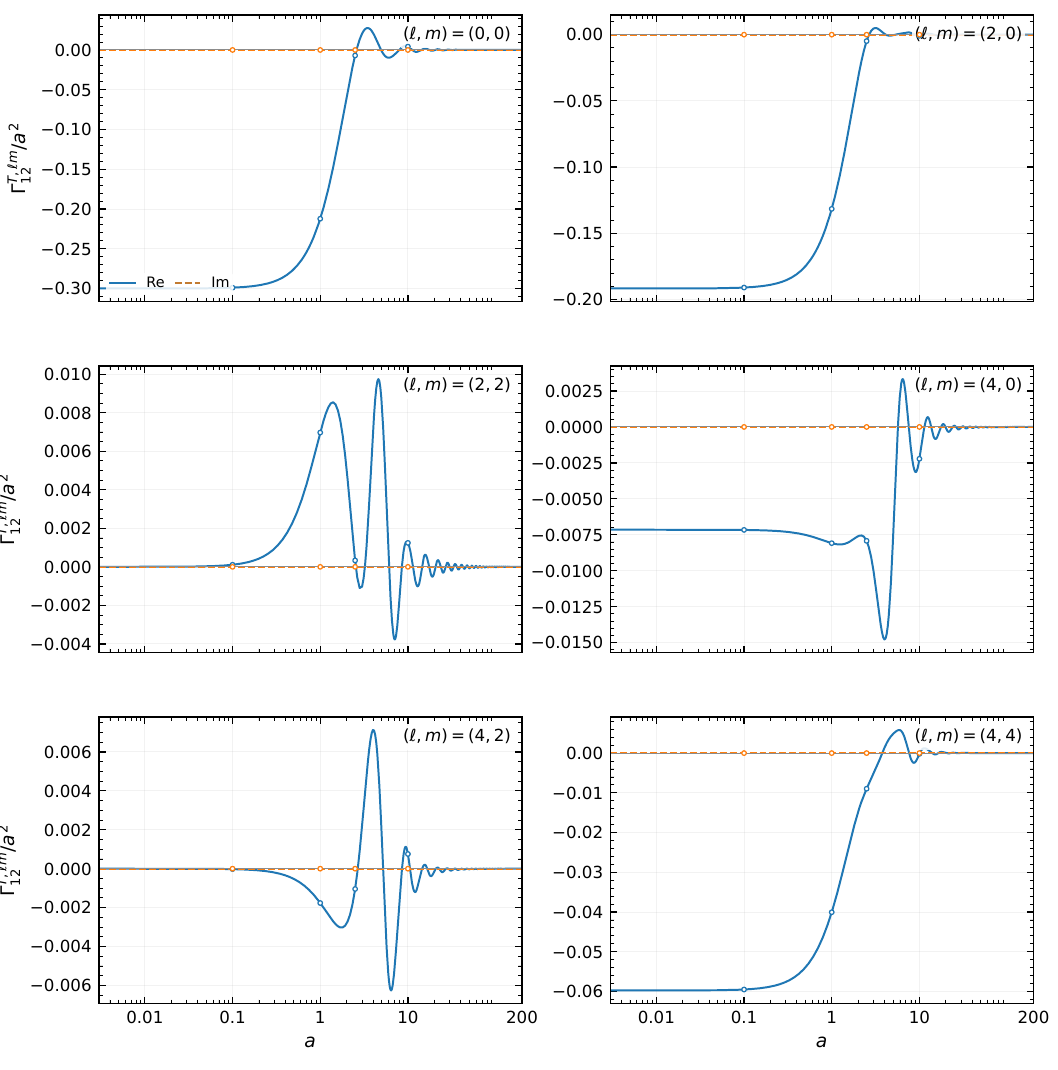}
\caption{All independent even tensor $12$ components through $\ell=5$. Their imaginary parts vanish by symmetry. Signs and zeros represent enhancement and cancellation in weighted channel correlations; they are not signs of the sky power.}
\label{fig:even}
\end{widefigure}

\subsection{Odd multipoles and directed cross correlations}
The odd cross-correlation components in Fig.~\ref{fig:odd} are purely imaginary and reverse sign when the channel order is exchanged. They arise from the direction-dependent phase difference between two readouts at finite light-travel time. The leading static angular kernel is even under $\bn\to-\bn$, so these components have no $a^2$ contribution. Equation~\eqref{eq:odd-low} shows that the first propagation correction is not the same order for every odd multipole.

For a real in-plane dipole sky, these kernels enter the quadrature component of the cross spectrum. In the time domain, they encode a directed lag correlation. Ordinary intensity anisotropy suffices to produce this structure; a circularly polarized background is not required. Retaining only the real part of the cross spectrum would therefore discard this channel pair's response to those intensity multipoles. A general complex delay combination need not preserve the real- and imaginary-part rules of $M_i$; its phase structure follows from the generation matrix.

The comparison between $h_{11}^T$ and $h_{33}^T$ is also useful for interpreting angular resolution. A kernel for a higher-order sky pattern can dominate a lower-order one over part of the band because the two channels sample different propagation phases. This is a response property, rather than a universal statement about which multipoles can be measured most precisely. A detectability comparison requires the corresponding sky amplitudes and the frequency-dependent noise weights.

\begin{widefigure}
\includegraphics[width=\textwidth]{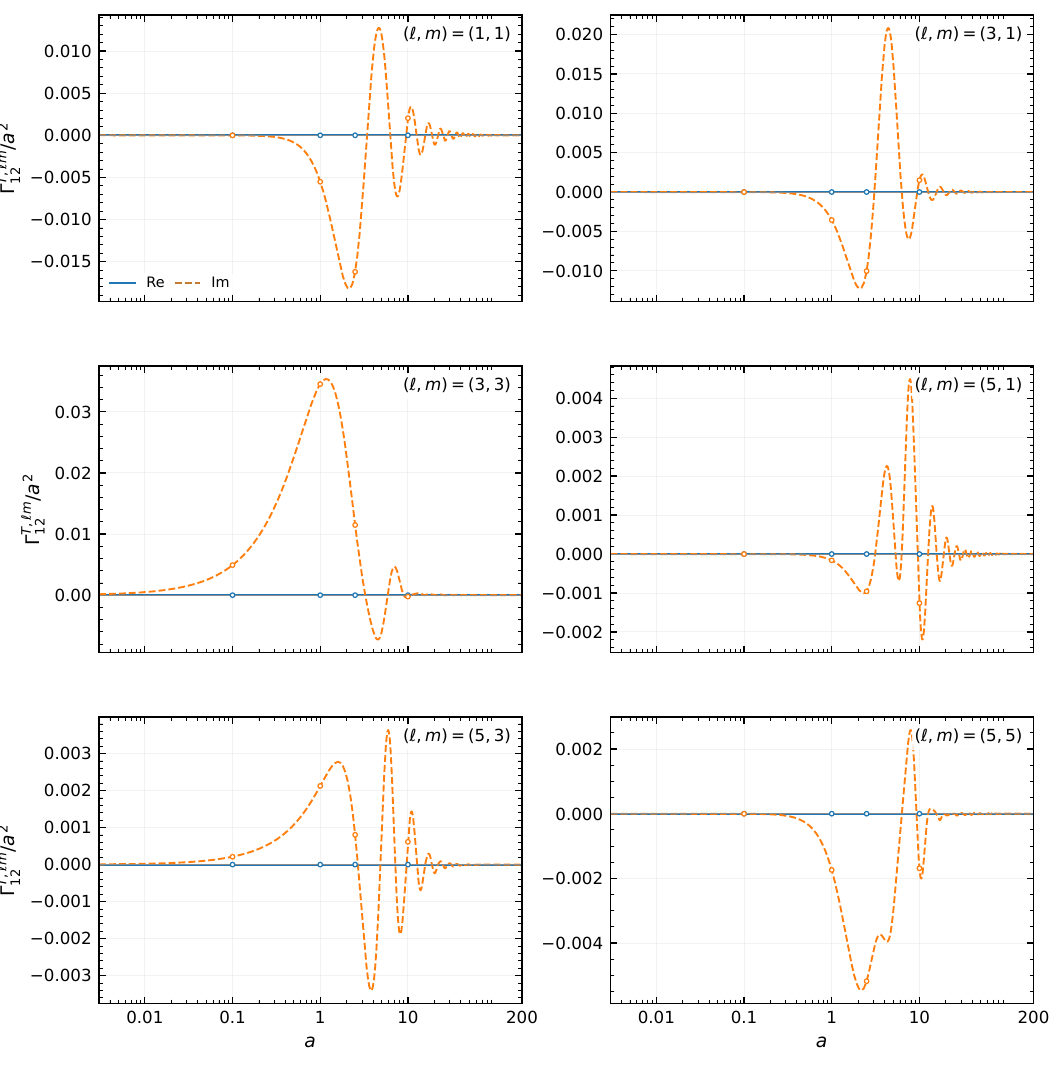}
\caption{All independent odd tensor $12$ components through $\ell=5$. The imaginary parts encode directed angular correlations, while the real parts vanish by symmetry. An analysis retaining only the real cross spectrum loses this channel pair's response to these intensity multipoles.}
\label{fig:odd}
\end{widefigure}

\subsection{Multipole norms and high-frequency tails}
At fixed angular order, define the rotationally invariant norm
\begin{equation}\label{eq:norm}
 \mathcal N_{CC',\ell}^{T}
 =\left[\sum_{m=-\ell}^{\ell}
       |\Gamma_{CC'}^{T,\ell m}|^2\right]^{1/2}.
\end{equation}
Figure~\ref{fig:norm} plots $\mathcal N_\ell^T/a^2$, including all positive and negative $m$. This norm summarizes the response within one angular order and removes the orientation dependence of an individual coordinate component. The $M_1$ auto response contains only even orders, while the $M_1M_2$ cross response retains odd orders. Their relative sizes can change with frequency, and there is no universal monotonic hierarchy in $\ell$. The norm is a measure of the detector kernel, distinct from a sky angular power spectrum $C_\ell$ or a noise-weighted sensitivity.

\begin{widefigure}
\includegraphics[width=\textwidth]{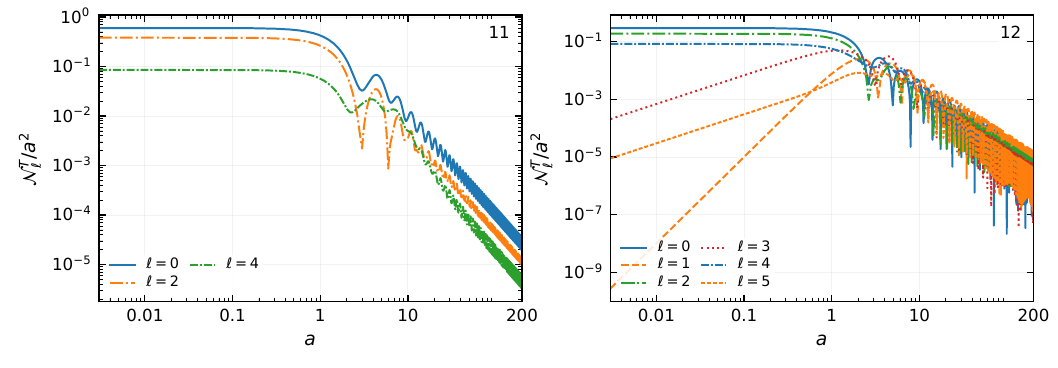}
\caption{Full-$m$ norms of the tensor auto- and cross-correlation multipoles. Odd cross responses are suppressed by additional propagation orders at low frequency. At finite frequency, the response amplitudes have no universal monotonic ordering in $\ell$.}
\label{fig:norm}
\end{widefigure}

Figure~\ref{fig:tail} resolves the high-frequency interval $40\leq a\leq200$. Near the vanishing link projection $\delta=1+\bn\cdot\bu=0$, the tensor directional numerator is $O(\delta^2)$, and $|q_a(\delta)|^2\leq\min(a^2,4/\delta^2)$. The angular norm of a single-link response is consequently bounded uniformly in $a$. A finite delay polynomial also has bounded coefficients on the unit circle. Applying Cauchy--Schwarz gives, for each fixed $\ell,m$,
\begin{equation}\label{eq:tensor-high-bound}
 |\Gamma_{CC'}^{T,\ell m}|=O(1),\qquad
 \mathcal N_{CC',\ell}^{T}/a^2=O(a^{-2}).
\end{equation}
This is an upper bound; particular channels or multipoles can decay faster because of cancellations. The oscillatory zeros in the plots are compatible with this envelope and do not imply a general cutoff of angular response at high frequency.

The extended linear-frequency panels complement the logarithmic-frequency overview. They separate the envelope from increasingly dense oscillations and expose the small values produced by cancellation. Stable evaluation in these regions requires absolute as well as relative error checks. Extending the curves beyond the representative design windows serves this mathematical and numerical purpose; a mission forecast must restrict the kernels to its usable band and combine them with its noise model.

\begin{widefigure}
\includegraphics[width=\textwidth]{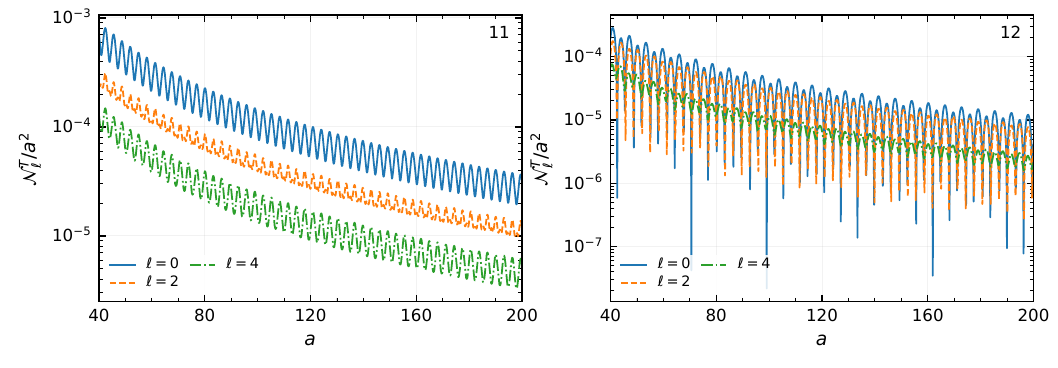}
\caption{High-frequency tensor multipole norms. A linear frequency axis resolves oscillations over $a=40$--$200$; the vertical axis is logarithmic. The extended interval shows cancellation zeros and the envelope that is less easily resolved in the main logarithmic-frequency plots.}
\label{fig:tail}
\end{widefigure}

\section{Representative TDI families, accuracy, and efficiency}

\subsection{Tensor responses of standard TDI and AET}
To illustrate the generation theorem for different combinations of optical paths, we use the first-generation Michelson $(X,Y,Z)$, Sagnac $(\alpha,\beta,\gamma)$, and Relay $(U,V,W)$ families~\cite{tdi}. AET is a common combination of the fundamental round-trip readouts; constructing it from $X,Y,Z$ additionally retains the common factor $1-z^2$. All families use the same link directions and the delay convention $z=e^{-\ii a}$. Appendix~\ref{app:standard-tdi} supplies six-link definitions and matrices for every ordered channel pair.

The short relations needed here are
\begin{align}\label{eq:standard-short}
 X&=(1-z^2)M_1,\nonumber\\
 \alpha&=(1-z)M_1-zS,\\
 U&=-zM_1-M_2-z^2M_3.\nonumber
\end{align}
Cyclic relabeling gives the remaining readouts. These examples respectively retain a common delay filter, a symmetric closed-loop contribution, and a combination of relative delays, all reusing the same $11/12$ functions.

\begin{widefigure}
\centering
\includegraphics[width=\textwidth]{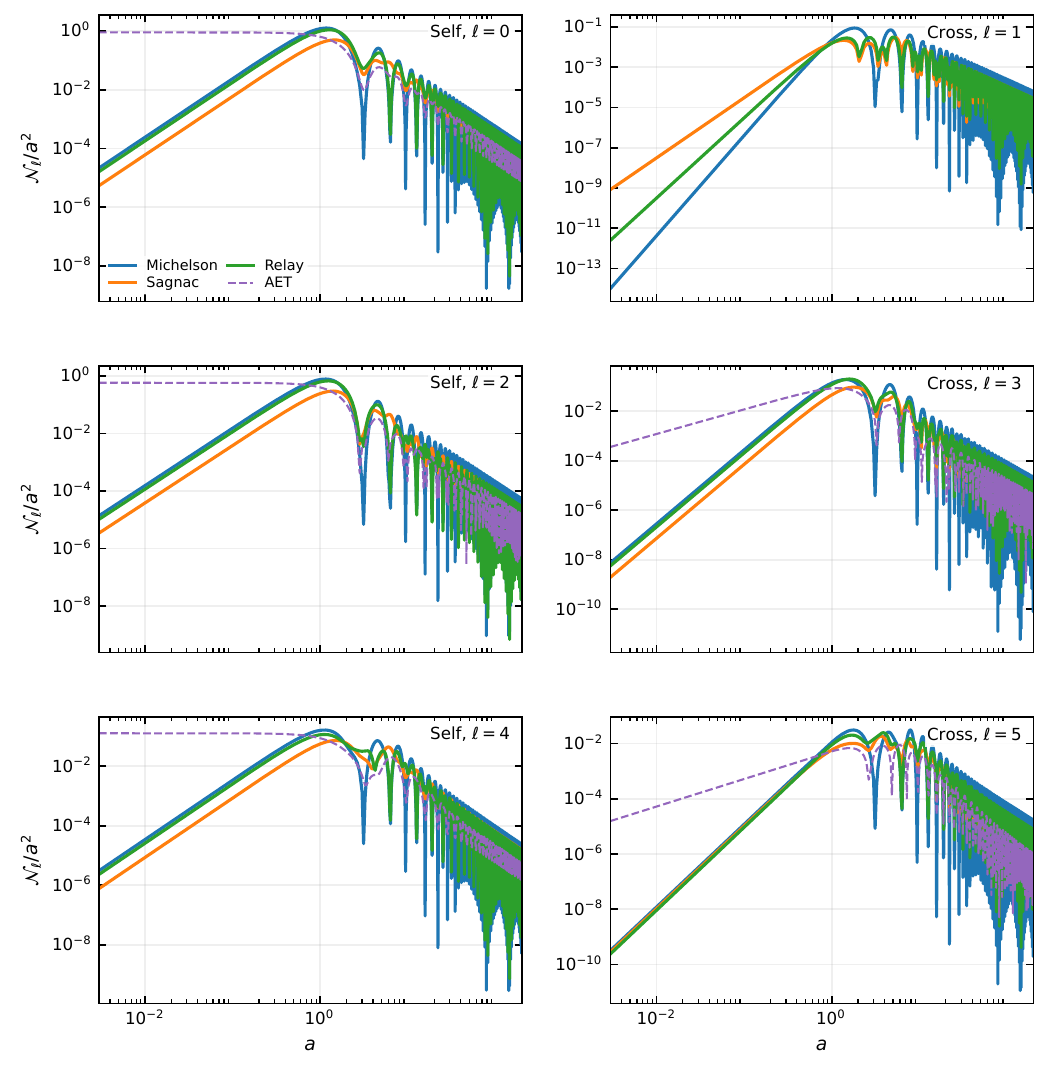}
\caption{Tensor multipole norms of three standard TDI families with their actual delay filters retained. Left: first-channel autocorrelations at $\ell=0,2,4$. Right: cross correlations of the first two channels at $\ell=1,3,5$. All positive and negative $m$ enter each norm, divided only by $a^2$. The interval is $0.003\leq a\leq200$; no noise or sky-power weighting is assigned. Purple dashed curves give the $M_i$-based AET reference: $AA$ on the left and $AE$ on the right. The $AE$ dipole is exactly zero and is omitted from the logarithmic ordinate.}
\label{fig:standard-norms}
\end{widefigure}
Figure~\ref{fig:standard-norms} compares oscillations, filtering zeros, and angular-order weights in common conventions. The leading monopole responses obey
\begin{equation}
 \Gamma_{XX}^{T,00}:\Gamma_{\alpha\alpha}^{T,00}:
 \Gamma_{UU}^{T,00}=4:1:3
 \quad\text{at common order }a^4.
\end{equation}
These ratios include the filters in the readout definitions and cannot be interpreted directly as sensitivity ratios. The additional suppression of the symmetric $S$ readout must be retained separately; it is not the leading behavior of an ordinary vertex Sagnac. Readout rank loss at a special frequency also differs from a removable pole of a representation, as verified directly in the appendix.

Every Michelson multipole matrix carries the common factor \(|1-z^2|^2=4\sin^2a\). Its zeros at \(a=n\pi\) therefore affect every angular order and originate in the readout filter, rather than a geometric selection of one sky multipole. The ordinary Sagnac family includes a loop response; its zeros also depend on cancellation between this contribution and the round-trip difference. Relay readouts combine different vertices and relative delays and do not share this common filter.

A full multipole norm is invariant under a rigid change of attitude, but changes under a readout transformation. The plotted norms describe the weighting of angular structure by each family and cannot directly rank their statistical information. When three independent readouts are retained, an invertible transformation of both signal and noise preserves Fisher information. Actual filtering degeneracies require separate treatment. This distinction determines how the differences in the plotted norms should be interpreted.

\subsection{Readout phases and directed cross spectra}
\begin{widefigure}
\centering
\includegraphics[width=\textwidth]{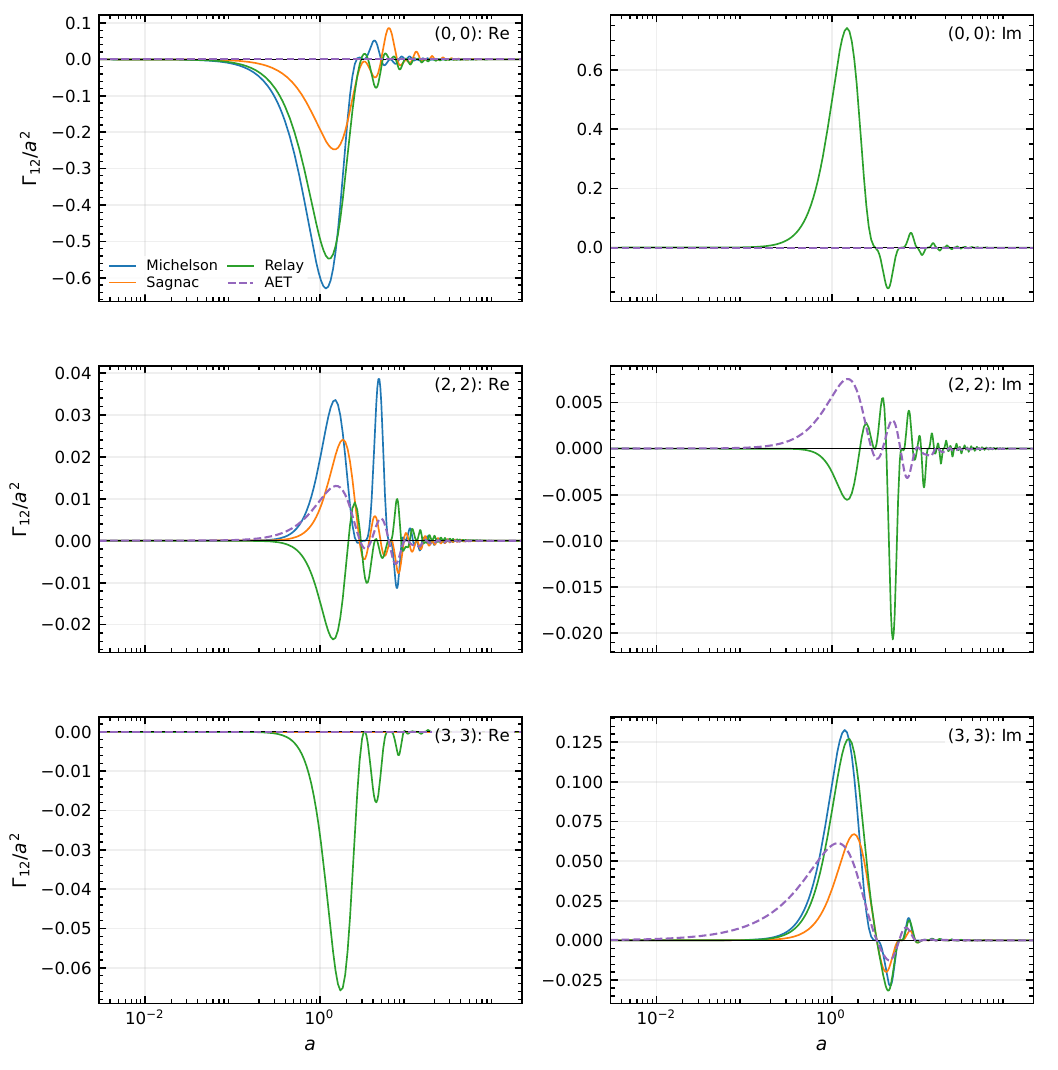}
\caption{Tensor cross multipoles for the first two readouts of each standard family. Rows select $(\ell,m)=(0,0),(2,2),(3,3)$; columns show real and imaginary parts of $\Gamma_{12}^{T,\ell m}/a^2$. Signs and delay phases are retained. Fundamental $M_1M_2$ quadrature selection rules are not imposed on general readouts. Purple dashed curves show the real or imaginary $AE$ component in the same convention. Its monopole is exactly zero; nonzero quadrupole and octupole components are retained.}
\label{fig:standard-phases}
\end{widefigure}
Sky weights and readout phases jointly determine the cross spectrum in Fig.~\ref{fig:standard-phases}. In particular, the Relay monopole satisfies
\begin{equation}\label{eq:relay-im-main}
 \Im\Gamma_{UV}^{T,00}
 =(2\sin a-\sin2a)(g_{00}^{T}-h_{00}^{T}).
\end{equation}
This imaginary part can be nonzero even for an isotropic sky. The purely imaginary odd multipoles and real even multipoles of the elementary round-trip pair therefore depend on that readout convention. An imaginary cross spectrum for a general TDI channel pair alone does not establish odd-multipole sky structure or circular polarization. Combining the full complex spectrum with a real sky intensity distribution yields a Hermitian covariance.

Cross-spectrum phases also depend on the relative time reference of the two readouts. A frequency-dependent rephasing multiplies their cross spectrum by the relative phase while leaving auto spectra unchanged. Comparisons therefore require fixed link directions, Fourier conventions, and delay definitions. The quadrature component has a statistical interpretation within the corresponding signal and noise model. Our plots retain the physical delay phases in these conventions, without rephasing channels to make their monopole cross spectrum real.

\subsection{Short formulas for general readouts}
For the monopole, $\varphi=\epsilon=1$, and Eq.~\eqref{eq:gen-S-elements} gives
\begin{equation}\label{eq:Smonopole}
 \Gamma_{SS}^{T,00}
 =\frac{3g_{00}^{T}+6h_{00}^{T}}{4\sin^2(a/2)}
 =\frac{a^6}{672}+O(a^8).
\end{equation}
Thus a channel not polynomially generated by the three $M_i$ is nonetheless expressed directly in terms of the same two elementary functions. Its low-frequency suppression and removable discrete limits are fixed by the complete numerator. The ordinary Sagnac response $\alpha_1$ follows similarly from Eq.~\eqref{eq:gen-alpha} and the generation matrix.

Angular selection for a general readout can differ from that of the elementary round-trip channels. For example, $C=M_1+zM_2$ obeys, for odd $\ell$,
\begin{align}\label{eq:arbitrary-odd-auto}
 \Gamma_{CC}^{T,\ell m}
 &=(z^*-z)h_{\ell m}^{T}=2\ii\sin(a)h_{\ell m}^{T},\\
 \Gamma_{CC}^{T,33}
 &=-\frac{\sqrt{35}}{60}a^4+O(a^6).\nonumber
\end{align}
The vanishing odd-$\ell$ elementary self function $g$ therefore does not forbid odd auto responses for arbitrary TDI channels. Figure~\ref{fig:arbitrary} shows the symmetric Sagnac monopole and the odd auto response of this delayed combination, illustrating two distinct uses of the same generation relation.

\begin{widefigure}
\includegraphics[width=\textwidth]{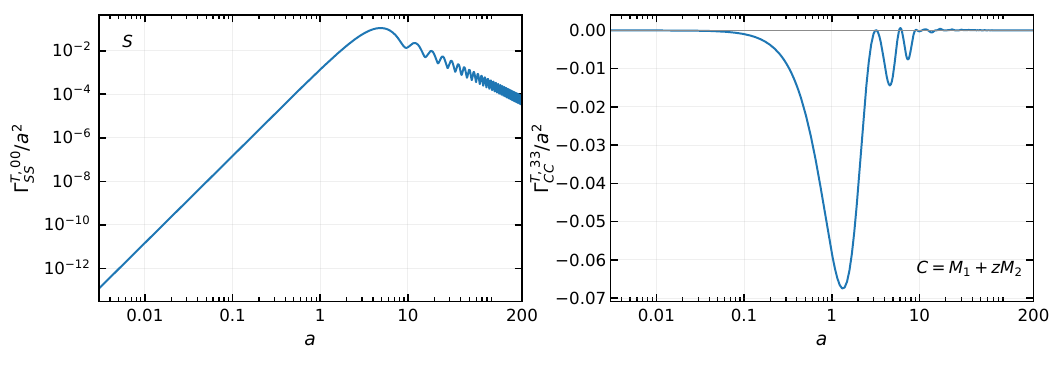}
\caption{Two analytically generated TDI examples. Left: the symmetric Sagnac tensor monopole $\Gamma_{SS}^{T,00}/a^2$, whose complete quotient is continuous at $a=2\pi n$. Right: the real odd auto response $\Gamma_{CC}^{T,33}/a^2$ for $C=M_1+zM_2$. Delay combinations can change the odd-multipole selection rule of elementary round-trip auto responses.}
\label{fig:arbitrary}
\end{widefigure}

Appendix~\ref{app:aet} gives the complete multipole formulas and cyclic-eigenchannel selection rules. Define the illustrative AET readouts by
\begin{equation}
 \begin{split}
 A&=(M_3-M_1)/\sqrt2,\\
 E&=(M_1-2M_2+M_3)/\sqrt6,\\
 \Tnull&=(M_1+M_2+M_3)/3.
 \end{split}
\end{equation}
Their monopole matrix is diagonal, with
\begin{align}
 \Gamma_{AA}^{T,00}
 &=\Gamma_{EE}^{T,00}=g_{00}^{T}-h_{00}^{T},\\
 \Gamma_{\Tnull\Tnull}^{T,00}
 &=\frac{g_{00}^{T}+2h_{00}^{T}}3
   =\frac{a^8}{6048}+O(a^{10}).\nonumber
\end{align}
Figure~\ref{fig:aet} displays representative $AA$, $EE$, and $\Tnull\Tnull$ auto responses together with $AE$, $A\Tnull$, and $E\Tnull$ cross responses. Orthogonality for an isotropic background does not diagonalize the covariance for a general sky; the physical cross spectrum must be assembled from the sky multipoles. We use the $1/3$ normalization for $\Tnull$; other null-channel conventions require the corresponding rescaling.

These examples place AET within the general construction. A new finite-delay channel is specified by its six link polynomials, checked against Eq.~\eqref{eq:noise}, and converted into the free basis using Eq.~\eqref{eq:gen-polyweights}. Its full multipole matrix then follows without an additional channel-pair sky integration. In applications, the noise covariance must be transformed together with the signal covariance before comparing the statistical performance of different readouts.

\begin{widefigure}
\includegraphics[width=\textwidth]{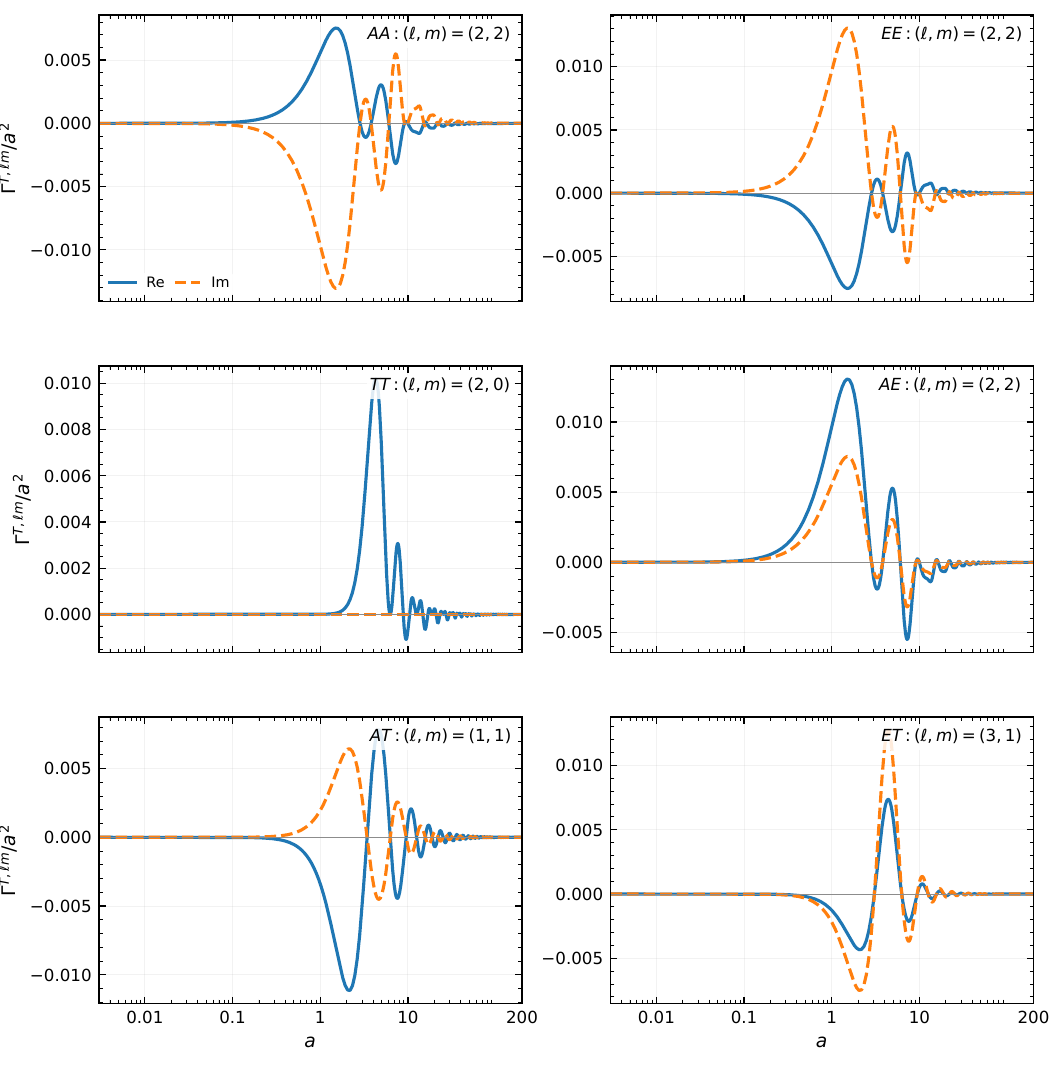}
\caption{Representative AET auto- and cross-response components generated from the $11/12$ matrix. Channel pairs and $(\ell,m)$ are marked inside the panels. In pair labels, $T$ denotes $\Tnull=(M_1+M_2+M_3)/3$; the superscript $T$ denotes tensor polarization. Blue solid and orange dashed curves show real and imaginary parts, respectively, with ordinate $\Gamma^{T,\ell m}/a^2$. The displayed $AA$ and $EE$ quadrupole components have opposite signs, whereas $\Tnull\Tnull$ allows only $m\equiv0\pmod3$. Nonzero cross responses show that monopole orthogonality does not diagonalize a general anisotropic sky.}
\label{fig:aet}
\end{widefigure}

\subsection{Accuracy and independent checks}
The elementary functions are tested using explicit Cartesian tensor contractions and the full sinc response in Eq.~\eqref{eq:q}. Sky integration uses Gauss--Legendre polar nodes and a uniform azimuthal grid. All tensor $11/12$ multipole components through $\ell=5$ are compared at eight positive and negative frequencies, yielding 576 comparisons with a maximum absolute residual of $1.38\times10^{-14}$. Zero frequency is evaluated by its exact limit. At very small nonzero frequency, high-precision finite expressions or Taylor expansions control cancellation.

An independent test of the generation theorem uses $S$, $\alpha_1$, and a polynomial channel generated with a fixed random seed. The analytic side reads only the final $g,h$ functions; the reference side combines the six complete light links directly. Eight frequencies cover $a=10^{-4}$, generic frequencies, $z=-1$, $z\approx1$, and the continuous limit at $z=1$. Each polarization class includes nine channel pairs and 36 multipoles, for 2592 comparisons. Across the four polarization classes, the 10368 comparisons have a maximum absolute residual of $6.18\times10^{-13}$; the maximum change under grid refinement is $3.21\times10^{-12}$. This test includes the Sagnac channel required for polynomial completeness and the discrete frequency limits.

The high-frequency extension uses dense linear sampling and a comparison of 45- and 85-digit evaluation. At $a=40,53,100,137,200$, 360 comparisons yield a maximum difference of $6.96\times10^{-48}$. Independent sky integration also tests eight tensor channel-multipole components at $a=200$. A $512\times2048$ grid gives a maximum absolute residual of $7.61\times10^{-14}$ and a maximum change of $2.53\times10^{-14}$ relative to a $384\times1536$ grid. These residuals measure agreement at the tested frequencies and grids; they are not a rigorous uniform error bound over the full frequency domain.

The standard-family checks directly assemble all six complete links, without using the generation matrix in the reference sky integral. We test $a=0.1,1,3,10,37$, all positive and negative $m$ through $\ell=5$, all nine ordered channel pairs, and $T,V,B,L$. Selected entries are also tested at $z=-1$, a nontrivial cube root, and nonzero frequencies with $z=1$. The 20,412 comparisons have maximum absolute residual $8.53\times10^{-14}$. At $a=37$, refinement from $96\times224$ to $144\times384$ changes the reference by at most $2.14\times10^{-13}$. All polarization classes use identical link definitions; this paper displays tensor results and Paper II's accompanying material supplies the others. Exact polynomial checks separately verify noise cancellation, family generation, and the special-point ranks. These tests establish agreement for the stated cases; they do not provide a uniform error bound over the full frequency range.

\subsection{Stable evaluation and matched-accuracy efficiency}
The cost of a finite analytic expression depends on its evaluation strategy and working precision. The generic evaluator sums expanded terms and prepares special functions for the entire library. The optimized implementation instead selects the dependencies of the requested responses and uses
\begin{equation}\label{eq:fast-e1}
 E_1(\ii r a)=-\Ci(|r|a)
 +\ii\operatorname{sgn}(r)[\Si(|r|a)-\pi/2]
\end{equation}
for \(a>0\), real \(r\ne0\), and the principal branch. Positive and negative imaginary arguments therefore share real sine- and cosine-integral values. A frequency-independent coefficient matrix then acts on shared phases, frequency powers, and special functions, eliminating repeated work across multipoles at one frequency.

At low frequency, we evaluate precomputed Taylor coefficients of \(\mathcal R=\Gamma/a^2\), through frequency degree 64. This frequency order is distinct from the angular range \(\ell\leq5\) of the explicit library. For the elementary round-trip channels, \(q_a(d)/a=\ii\int_0^1 e^{-\ii atd}\,\mathrm dt\) expresses the complete kernel as a finite sum of averaged plane waves. Each dimensionless path phase lies in \([-1,4]\), so the absolute phase difference is bounded by 5. Polarization-tensor norms and the spherical-harmonic addition theorem then imply the real-frequency truncation bound
\begin{equation}\label{eq:fast-remainder}
 |\mathcal R_{ij}^{P,\ell m}-\mathcal R_{ij,N}^{P,\ell m}|
 \leq B_P\sqrt{2\ell+1}\,
 \frac{(5|a|)^{N+1}}{(N+1)!},
\end{equation}
with \(i,j=1,2\), \(B_T=8\), \(B_V=4\), and \(B_B=B_L=2\). This bounds Taylor truncation, rather than certifying the entire floating-point implementation. The evaluator selects the required frequency order and estimates coefficient-rounding and summation errors. Strong cancellation in a closed expression or channel transformation triggers arbitrary-precision evaluation with a precision-refinement check. Special-function and assembly indicators still require independent numerical validation.

\subsection{Optimized computation at matched accuracy}
All conventional-accuracy timings require
\begin{equation}\label{eq:benchmark-tol}
 |\mathcal R-\mathcal R_{\rm ref}|\leq10^{-10}+10^{-8}|\mathcal R_{\rm ref}|,
\end{equation}
with references evaluated from the original finite expressions at 120 decimal digits. Both implementations are optimized. The direct sky integral precomputes nodes, weights, geometric projections, polarization contractions, and harmonic samples; it vectorizes the complete sinc response and shares the $11/12$ kernels across multipoles. Gauss--Legendre polar and periodic azimuthal grids are refined until both the reference error and the successive-grid change pass Eq.~\eqref{eq:benchmark-tol}. The analytic implementation prepares the requested component indices, symmetry maps, and shared coefficient data. Its frequency sweep evaluates phase and special-function arrays jointly. Shared power matrices were faster than Python-level Horner loops for the tested batches and are retained here.

The timings use one numerical linear-algebra thread on an AMD EPYC 9V74, with Python 3.12 and NumPy 2.3.5. Table~\ref{tab:efficiency} reports the median of seven repetitions after preparation; final ORF values are never cached. The previous and new analytic implementations are measured in this same run. Generic termwise multiprecision timings are retained separately in the reproducibility records as implementation diagnostics, rather than used as the main comparison against double-precision quadrature.

\begin{localtable}
\caption{Prepared evaluation times in milliseconds, at the matched tolerance of Eq.~\eqref{eq:benchmark-tol}. Old and New denote the previous optimized and new prepared analytic implementations. Q is optimized direct sky quadrature. A single entry is $g_{00}^T$; a batch consists of the 18 independent tensor components through $\ell=5$.}
\label{tab:efficiency}
\begin{ruledtabular}\begin{tabular}{ccrrrr}
$a$ & Entries & Old & New & Q & Q/New\\
0.003 & 1 & 0.035 & 0.023 & 1.23 & 54\\
0.003 & 18 & 0.052 & 0.033 & 1.41 & 43\\
1 & 1 & 0.041 & 0.024 & 1.40 & 58\\
1 & 18 & 0.054 & 0.031 & 1.47 & 48\\
10 & 1 & 0.033 & 0.020 & 1.57 & 79\\
10 & 18 & 0.045 & 0.030 & 2.22 & 73\\
200 & 1 & 0.029 & 0.022 & 115.06 & 5332\\
200 & 18 & 0.062 & 0.023 & 115.75 & 4948\\
\end{tabular}\end{ruledtabular}
\end{localtable}

The 18-component speedup is approximately $\SpeedOne$ at $a=1$ and $\SpeedHigh$ at $a=200$. A 32-frequency sweep over $0.003\leq a\leq200$, with 18 components per frequency, takes $\SweepMilliseconds\,\mathrm{ms}$, compared with $\QuadSweepMilliseconds\,\mathrm{ms}$ for quadrature on individually validated grids. The new sweep is $\SweepImprovement$ times faster than the previous analytic sweep. These are response-kernel costs. Both methods may reuse channel matrices, and their reuse is not counted as an advantage unique to analytic integration.

\subsection{Multiprecision accuracy and time to target}
The new multiprecision path retains exact source coefficients and returns multiprecision values throughout. It computes only the requested special functions, shares the positive and negative imaginary-$E_1$ arguments by conjugation, and reuses phases and powers across components. Precision-specific constants are prepared once. A separate comparison at equal working precision isolates expression organization: the 18-component implementation is approximately 2.2--3.4 times faster than the generic evaluator in the tested 60- and 120-digit cases. Working precision alone is not an accuracy guarantee.

The measurements use mpmath 1.3.0 with the gmpy2 2.3.2 arithmetic backend. For a direct accuracy--cost comparison, define the dimensionless mixed target
\begin{equation}\label{eq:mp-tolerance}
 |\mathcal R-\mathcal R_{\rm ref}|\leq10^{-p}(1+|\mathcal R_{\rm ref}|).
\end{equation}
We test $p=16,30,50,100$ at $a=0.003,1,10,200$ for $g_{00}^T$. The analytic evaluator compares two working precisions with a margin of $1/100$ of the target; it adds protection against low-frequency cancellation. Independent multiprecision sky quadrature uses full round-trip responses, invariant tensor contractions, reflection/inversion symmetries, and precomputed nodes. Its polar and azimuthal resolutions are increased together. Both the error against an independently implemented finite-expression reference and the successive-grid change must pass the target. All 16 tensor cases converged within the 90-second grid-search budget. The $p=50$, $a=1$ tests also pass for $V,B,L$; complete errors and grid histories are included in the supporting data. Equation~\eqref{eq:mp-tolerance} controls absolute error near zeros, and $p$ is not a claim of $p$ relative digits for a vanishing response.

\begin{table*}[t]
\caption{Multiprecision $g_{00}^T$ timings in milliseconds. First-call times include symbolic coefficient preparation and precision refinement (A), or node generation and grid refinement (Q); imports and construction of verification references are excluded. Repeated times reuse constants or nodes: A is the median of five calls, Q of three. The ratio uses repeated times. Both methods satisfy Eq.~\eqref{eq:mp-tolerance}; they need not use the same working precision.}
\label{tab:mp-efficiency}
\begin{ruledtabular}\begin{tabular}{ccrrrrr}
$a$ & $p$ & A first & Q first & A repeated & Q repeated & Q/A\\
1 & 30 & 8.91 & 131.0 & 0.808 & 51.9 & 64\\
1 & 100 & 9.67 & 745.7 & 1.072 & 242.0 & 226\\
10 & 50 & 11.21 & 1448.3 & 1.273 & 471.1 & 370\\
10 & 100 & 11.50 & 3205.8 & 1.601 & 1039.7 & 649\\
200 & 50 & 10.98 & 21064.1 & 2.200 & 6783.4 & 3083\\
200 & 100 & 17.79 & 48397.4 & 4.105 & 16711.5 & 4071\\
\end{tabular}\end{ruledtabular}
\end{table*}

\begin{widefigure}
\includegraphics[width=\textwidth]{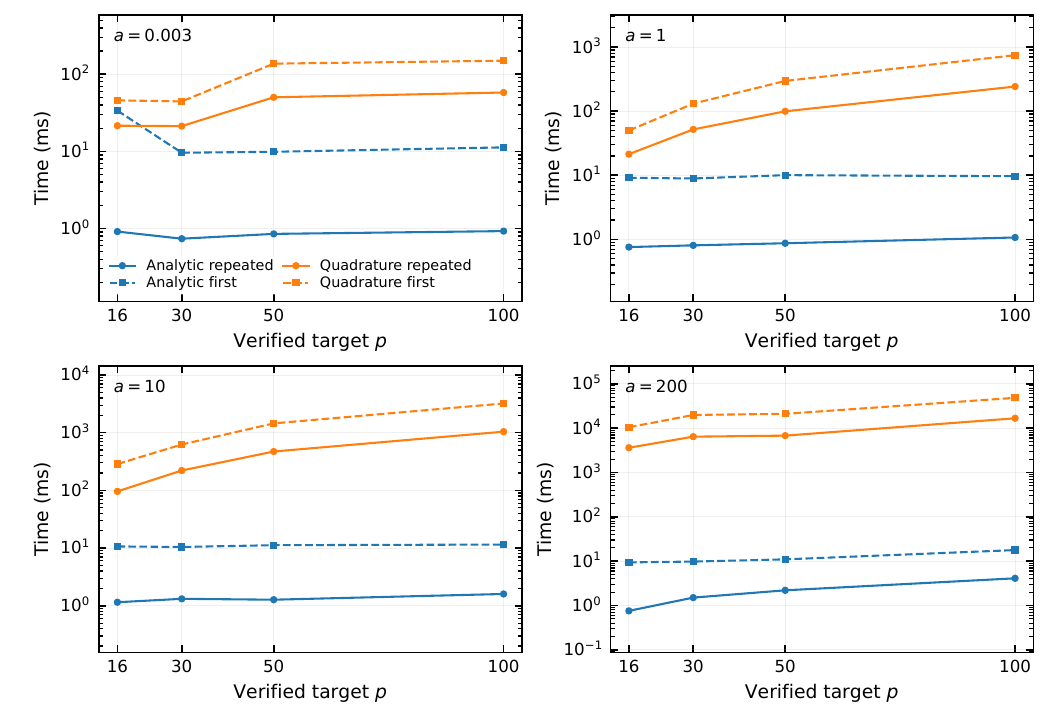}
\caption{Cost versus verified precision target for the tensor monopole $g_{00}^T$. Every point passes Eq.~\eqref{eq:mp-tolerance}; actual errors can be smaller than requested and are recorded separately. Blue curves denote analytic evaluation and orange curves optimized multiprecision sky quadrature. Solid curves show prepared repeated evaluation, dashed curves first-call preparation and refinement. The panels share the physical coordinate-frame dimensions used in Fig.~\ref{fig:norm}. These benchmarks quantify this implementation, not an optimum over all quadrature algorithms.}
\label{fig:precision-cost}
\end{widefigure}

The independent tests demonstrate precision extensibility beyond machine doubles. They do not imply that selecting a working precision automatically delivers that many accurate digits, or that arbitrary-precision quadrature cannot reach high accuracy. Analytic evaluation removes angular discretization, while still requiring accurate special functions and cancellation control. Figure~\ref{fig:precision-cost} measures the cost of meeting prescribed targets. The measured gain depends on frequency, target, and preparation reuse and need not increase monotonically with $p$.

Removable delay points admit an additional local strategy. For the normalized Sagnac monopole set $Q(a)=3[g_{00}^P(a)+2h_{00}^P(a)]/a^2$ and $a_n=2\pi n\ne0$. Writing $Q(a_n+\delta)=\sum_{j\geq2}q_j\delta^j$ removes its known double zero before division:
\begin{equation}\label{eq:local-sagnac}
 \frac{\Gamma_{SS}^{P,00}(a_n+\delta)}{(a_n+\delta)^2}
 =\frac{\sum_{j\geq0}q_{j+2}\delta^j}{\sinc^2(\delta/2)}.
\end{equation}
The coefficients follow from analytic Taylor arithmetic on the finite formulas. This avoids subtracting two nearly cancelling responses and dividing by $|z-1|^2$. The implementation tests degree refinement; 16 checks at $\delta=0,\pm10^{-4},10^{-12}$ for $T,V,B,L$ agree with the original continuous-limit evaluator at the $p=100$ target. Local preparation is separate from subsequent evaluation and is useful for repeated nearby frequencies.

The new frequency evaluator passes 10,384 comparisons against 150-digit references, including negative frequencies, negative $m$, symmetry zeros, and $a=0$; its largest error divided by Eq.~\eqref{eq:benchmark-tol} is $5.98\times10^{-4}$. The new multiprecision path passes 1,232 checks at $p=30,100$, including $a=10^{-8}$; the largest error-to-target ratio is $7.46\times10^{-37}$. Earlier independent full-link tests remain archived with their original scopes. Precision-refinement indicators and these finite test sets are not uniform interval certificates. Ill-conditioned downstream inference still requires precision chosen for its own error propagation.
\FloatBarrier

\section{Conclusions and outlook}
We have combined delay-polynomial channel algebra, constellation symmetries, and the directional-projection integration method of Zhou and Li~\cite{zhouli,source} to construct anisotropic TDI ORFs systematically. The problem is to generate complete finite-propagation responses for arbitrary admissible finite-delay readouts without repeating a sky integral for every channel pair. Starting from the laser-noise cancellation identities, we establish the free basis $(M_1,M_2,S)$ of the full polynomial channel module. The resulting construction separates channel coefficients from elementary angular kernels, so a change of readout becomes an algebraic transformation of reusable response data. This connects the established algebraic description of TDI and geometric response relations~\cite{dhurandhar02,kudoh05} to a complete channel-generation and evaluation procedure.

For the tensor, vector, breathing, and longitudinal intensity sectors, rotation and inversion reduce the required angular data to the two families $g_{\ell m}^P$ and $h_{\ell m}^P$. The channel algebra itself does not depend on polarization. The representative formulas and figures in this paper describe tensor intensity; other polarization contractions enter through their corresponding elementary kernels. The method applies at arbitrary fixed finite angular order, whereas the distributed explicit library is tabulated through $\ell=5$. Polynomial completeness, single-frequency rank, and continuity of the physical response remain distinct: rational channel coordinates can have poles even though the full response is finite. Taking the relevant combinations and limits jointly provides both the mathematical interpretation and a stable evaluation strategy.

The construction yields measurable computational gains. At the matched tolerance of Eq.~\eqref{eq:benchmark-tol}, the prepared batch of 18 components is approximately $\SpeedOne$ times faster than optimized direct sky quadrature at $a=1$ and $\SpeedHigh$ times faster at $a=200$ in the benchmark considered here. A 32-frequency sweep takes $\SweepMilliseconds\,\mathrm{ms}$, compared with $\QuadSweepMilliseconds\,\mathrm{ms}$ for quadrature, and is $\SweepImprovement$ times faster than the previous analytic implementation measured in the same run. Multiprecision expressions share exact coefficients and special functions; independent full-response monopole integrations verify the prescribed mixed error targets through $p=100$. At $a=10$, the $p=100$ repeated evaluation is approximately $\MPMidGain$ times faster than the tested multiprecision quadrature. These results concern the stated kernels, implementations, and tolerances. They are empirical accuracy and cost measurements, rather than uniform error certificates or end-to-end inference speedups.

The reusable angular kernels make multichannel covariance construction and readout comparisons practical over broad frequency and angular ranges. Finite propagation phases expose multipoles suppressed at long wavelength, alter cross-spectrum phases, and produce transfer zeros. Anisotropic AET cross responses demonstrate that isotropic signal orthogonality does not imply a diagonal signal covariance for a general sky. The localized treatment of removable delay points and adjustable arithmetic precision also support the study of weak responses and cancellations. These capabilities provide computational inputs for subsequent sky reconstruction and noise-aware inference; sensitivity and identifiability additionally depend on sky spectra, instrumental noise, observing time, scanning, and calibration.

The present framework assumes a static, reciprocal, equal-arm equilateral constellation and finitely many commuting delays. Extending it to unequal arms and moving constellations requires an appropriate channel-generation structure and additional angular kernels. Higher angular orders are another computational extension, with expression size and conditioning to be measured explicitly. The range up to $a=200$ tests finite-propagation structure and evaluation cost, rather than prescribing a common mission science band. The ET scale comparison in Sec.~V motivates a separate ground-based extension: ideal resonant Fabry--Perot responses admit a sky-independent optical factor~\cite{rakhmanov2008}, but an ET response calculation requires a separately validated optical-readout mapping and normalization; sensitivity studies additionally require its noise model~\cite{etdesign2020,etscience2023,etgeometry2026,virtuoso2025}.

\begin{acknowledgments}
This work was supported by the National Natural Science
Foundation of China (Grant Nos. 12575072,
12547101 and 125B2102), the Fundamental Research
Funds for the Central Universities Project (Grant No.
2024IAIS-ZD009), the National Key Re-
search and Development Program of China (Grant No.
2021YFC2203004).
\end{acknowledgments}

\section*{Data and code availability}
No observational data were used. The accompanying collection~\cite{TDIFormulaLibrary} contains the final analytic expressions, a formula-viewing Jupyter notebook, figure data, and numerical validation and timing records. The explicit tables cover the 11/12 intensity coefficients of the $T,V,B,L$ sectors through $\ell=5$, in the conventions of this paper. The public release is limited to these formulas and supporting numerical data; the authors' numerical implementations are not publicly released. Access for verification may be arranged with the corresponding author upon reasonable request. The timing results describe the implementations tested in this paper, rather than the performance of distributed software. Repository details remain to be completed before submission: \FormulaRepositoryCitation.

\appendix
\section{Representative TDI channels: transformations and multipole responses}
\label{app:standard-tdi}
We translate three standard first-generation readout families into the six-link conventions of this paper. Static, reciprocal, equal arms are assumed. Link directions and Fourier phases must be converted before importing another convention. The standard families and their physical interpretations are reviewed in Ref.~\cite{tdi}.

\subsection{Common conventions and generation relation}
For a channel vector $\boldsymbol C=\mathsf C\boldsymbol M$, the multipole response follows from
\begin{equation}\label{eq:standard-generation}
 \boldsymbol\Gamma_{\rm family}^{P,\ell m}
 =\mathsf C\mathsf K_{\ell m}^{P}\mathsf C^\dagger.
\end{equation}
\subsection{Definitions and transformations of standard TDI}
\subsubsection{Definitions and delay coefficients}
Let $\sigma$ denote $1\to2\to3\to1$. We define
\begin{align}\label{eq:standard-defs}
 X&=(1-z^2)M_1,\quad Y=\sigma X,\quad Z=\sigma^2X,\nonumber\\
 \alpha&=y_{21}+zy_{32}+z^2y_{13}
          -y_{31}-zy_{23}-z^2y_{12},\\
 \beta&=\sigma\alpha,\quad\gamma=\sigma^2\alpha,\nonumber\\
 U&=z\gamma-\beta,\quad V=\sigma U,\quad W=\sigma^2U.\nonumber
\end{align}
The common delay filter of $X,Y,Z$ is retained; these are not identified with the unfiltered $M_i$. The symmetric readout $S$ uses the convention of the main text with its common delay removed.

In the order $(12,21,13,31,23,32)$, representative coefficient rows are
\begin{align}
 \boldsymbol d_X&=(1-z^2)(z,1,-z,-1,0,0),\nonumber\\
 \boldsymbol d_\alpha&=(-z^2,1,z^2,-1,-z,z),\\
 \boldsymbol d_U&=(1-z^2,0,0,z-z^3,z^2-z,z^3-1).\nonumber
\end{align}
Cyclic relabeling gives the other rows. Direct substitution yields polynomial identities $\boldsymbol d_C\mathsf B=0$, rather than cancellation at only one frequency. Since $M_1+M_2+M_3=(z-1)S$,
\begin{align}\label{eq:standard-basis}
 \alpha&=(1-z)M_1-zS,\nonumber\\
 U&=(z^2-z)M_1+(z^2-1)M_2+(z^2-z^3)S\\
  &=-zM_1-M_2-z^2M_3.\nonumber
\end{align}
The final relation compresses the relative link-path delays into polynomial weights on the three fundamental readouts.

\subsubsection{Complete multipole matrices}
With $\boldsymbol M=(M_1,M_2,M_3)^{\mathsf T}$ and $\mathsf J_{ij}=1$, write each family as $\boldsymbol C=\mathsf C\boldsymbol M$ for $z\ne1$, where
\begin{align}\label{eq:standard-matrices}
 \mathsf C_X&=(1-z^2)\mathsf I,\nonumber\\
 \mathsf C_\alpha&=(1-z)\mathsf I-\frac{z}{z-1}\mathsf J,\\
 \mathsf C_U&=\begin{pmatrix}-z&-1&-z^2\\-z^2&-z&-1\\-1&-z^2&-z\end{pmatrix}.\nonumber
\end{align}
Every auto- and cross-correlation multipole follows from the common generation relation in Eq.~\eqref{eq:standard-generation}.
Writing $q=e^{2\pi\ii m/3}$, we have $z^*=z^{-1}$ and $q^*=q^2$ on the unit circle. Each entry is a finite linear combination of $g$ and $h$. The electronic library lists all nine ordered channel pairs for both $\epsilon=\pm1$. Complex delays do not generally preserve the fundamental relation $\Gamma_{ji}=(-1)^\ell\Gamma_{ij}$. The applicable general identity is
\begin{equation}
 \Gamma_{ji}^{P,\ell,-m}=(-1)^m
       (\Gamma_{ij}^{P,\ell m})^*.
\end{equation}
A Hermitian physical covariance follows after combining the multipoles with a real sky intensity distribution. Planar reflection still forces components with odd $\ell+m$ to vanish.

The Relay monopole provides a short phase check. With $b=1+z+z^2$,
\begin{align}\label{eq:relay-monopole}
 \Gamma_{UU}^{00}&=3g_{00}+(|b|^2-3)h_{00},\nonumber\\
 \Gamma_{UV}^{00}&=(2z^{-1}+z^2)(g_{00}-h_{00})+|b|^2h_{00},\\
 \Im\Gamma_{UV}^{00}&=(2\sin a-\sin2a)(g_{00}-h_{00}).\nonumber
\end{align}
Although the monopole $g,h$ are real, a general TDI cross spectrum can acquire an imaginary part from its readout phase. These identities apply to every intensity class in our model.

\subsection{AET multipoles and selection rules}
\label{app:aet}
This subsection applies the two-function theorem to the AET readouts.
We use $A=(M_3-M_1)/\sqrt2$, $E=(M_1-2M_2+M_3)/\sqrt6$,
and $\Tnull=(M_1+M_2+M_3)/3$. For each intensity sector considered here,
\begin{equation}\label{eq:aet-matrix}
 \boldsymbol\Gamma_{A E\Tnull}^{P,\ell m}
 =\mathsf O\mathsf K_{\ell m}^{P}\mathsf O^{\mathsf T},\quad
 \mathsf O=\begin{pmatrix}
 -1/\sqrt2&0&1/\sqrt2\\
 1/\sqrt6&-2/\sqrt6&1/\sqrt6\\
 1/3&1/3&1/3
 \end{pmatrix}.
\end{equation}
The last row is not normalized to unity. The standard orthonormal channel is
$T_o=(M_1+M_2+M_3)/\sqrt3=\sqrt3\,\Tnull$.
Replacing one $\Tnull$ index by $T_o$ multiplies the response by $\sqrt3$;
replacing both multiplies it by 3.
Write $g=g_{\ell m}^{P}$, $h=h_{\ell m}^{P}$,
$q=e^{2\pi\ii m/3}$, and $s_m=1+q+q^2$.

For even $\ell$, the six ordered upper-triangular entries are
\begin{align}\label{eq:aet-even}
 \Gamma_{AA}&=\tfrac12(1+q^2)g-q^2h,\nonumber\\
 \Gamma_{EE}&=\chi_m\Gamma_{AA},\nonumber\\
 \Gamma_{\Tnull\Tnull}&=\tfrac{s_m}{9}(g+2h),\nonumber\\
 \Gamma_{AE}&=\frac{(q^2-1)g+2(1-q)h}{\sqrt{12}},\\
 \Gamma_{A\Tnull}&=\frac{(q^2-1)g+(q-1)h}{3\sqrt2},\nonumber\\
 \Gamma_{E\Tnull}&=\frac{(1-2q+q^2)g-(1+q-2q^2)h}{3\sqrt6},\nonumber
\end{align}
where $\chi_m=1$ for $m\equiv0\pmod3$ and $-1$ otherwise.
For odd $\ell$, $g=0$ and all three diagonal entries vanish, whereas
\begin{align}\label{eq:aet-odd}
 \Gamma_{AE}&=\tfrac{s_m}{\sqrt3}h,\nonumber\\
 \Gamma_{A\Tnull}&=\tfrac{2q^2-1-q}{3\sqrt2}h,\\
 \Gamma_{E\Tnull}&=\tfrac{1-q}{\sqrt6}h.\nonumber
\end{align}
The lower triangle follows from $\Gamma_{O'O}^{\ell m}=(-1)^\ell\Gamma_{OO'}^{\ell m}$.
An individual complex-harmonic matrix is not required to be Hermitian;
Hermiticity applies to the covariance obtained by combining it with a real sky intensity distribution.
All intensity entries vanish when $\ell+m$ is odd.

\subsubsection{Cyclic eigenchannels and angular selection}
The cyclic permutation acts on the standard orthonormal channels as
\begin{align}
 \sigma A&=-\tfrac12A+\tfrac{\sqrt3}{2}E,\nonumber\\
 \sigma E&=-\tfrac{\sqrt3}{2}A-\tfrac12E,\qquad\sigma T_o=T_o.
\end{align}
The complex channels $v_\pm=(A\pm\ii E)/\sqrt2$ have eigenvalues
$\lambda_\pm=e^{\mp2\pi\ii/3}$, with $\lambda_{T_o}=1$.
Rotational covariance requires
\begin{equation}\label{eq:aet-character}
 [q-\lambda_C\lambda_{C'}^*]\Gamma_{CC'}^{\ell m}=0.
\end{equation}
The two $v_\pm$ auto responses and the $T_o$ auto response therefore allow only $m\equiv0\pmod3$.
The $v_+v_-$ and $v_-v_+$ entries allow $m\equiv1$ and 2, respectively.
For even $\ell$, $\Gamma_{EA}=\Gamma_{AE}$, so
$\Gamma_{v_+v_+}=(\Gamma_{AA}+\Gamma_{EE})/2$:
the cross terms cancel without requiring $\Gamma_{AE}=0$.
This gives $\Gamma_{EE}=-\Gamma_{AA}$ when $m\not\equiv0\pmod3$;
the explicit formulas give equality when $m\equiv0\pmod3$.

The AET channel space decomposes into one trivial representation and one real two-dimensional representation.
The complex channels $v_\pm,T_o$ diagonalize the cyclic operation.
They do not diagonalize every anisotropic multipole matrix, which carries a nontrivial rotational weight for general $m$.
In particular, $\Gamma_{AE}$, $\Gamma_{A\Tnull}$, and $\Gamma_{E\Tnull}$ can be nonzero.
The intensity monopole and a noise covariance satisfying the standard symmetry assumptions are diagonal in AET~\cite{bartolo}.

The AE response vanishes for every dipole component, and both AT and ET vanish when $m\equiv0\pmod3$. For the other residue classes,
\begin{equation}
 \Gamma_{E\Tnull}^{\ell m}=\frac{2q+1}{\sqrt3}\Gamma_{A\Tnull}^{\ell m}.
\end{equation}
The factor is $\pm\ii$, so these two cross responses have equal multipole norms. Related geometric relations in the tensor sector were established in Ref.~\cite{kudoh05}.

For $\mathcal N_{OO',\ell}=[\sum_m|\Gamma_{OO'}^{\ell m}|^2]^{1/2}$,
\begin{align}
 \mathcal N_{AA,\ell}&=\mathcal N_{EE,\ell},\nonumber\\
 \mathcal N_{AE,\ell}^2&=3\!\sum_{m\equiv0\ (3)}|h_{\ell m}^{P}|^2,
 \quad\ell\ {\rm odd},\\
 \mathcal N_{\Tnull\Tnull,\ell}^2
 &=\tfrac19\!\sum_{m\equiv0\ (3)}|g_{\ell m}^{P}+2h_{\ell m}^{P}|^2,
 \quad\ell\ {\rm even}.\nonumber
\end{align}
These are applications of the general channel theorem; the AET choice does not restrict the admissible readouts.

\subsection{Special frequencies and actual rank loss}
The determinants are
\begin{align}
 \det\mathsf C_X&=(1-z^2)^3,\nonumber\\
 \det\mathsf C_\alpha&=-(z-1)(z^2+z+1),\\
 \det\mathsf C_U&=-(z-1)^2(z^2+z+1)^2.\nonumber
\end{align}
The rational Sagnac representation has an apparent pole at $z=1$, but its six-link definition is finite and all three rows approach $-S$. This does not imply a divergent GW response. An invertible-basis argument also cannot be applied at a true transfer zero. Table~\ref{tab:standard-ranks} gives the rank of the six-link coefficient rows, not a polarization Fisher rank or a guarantee of response to every sky.
\begin{localtable}
\caption{Six-link channel ranks at specified delays. The other nontrivial cube root has the same ranks as $\omega_3=e^{2\pi\ii/3}$.}
\begin{ruledtabular}
\begin{tabular}{lccc}
Family & $z=1$ & $z=-1$ & $z=\omega_3$\\
Michelson &0&0&3\\
Sagnac &1&3&2\\
Relay &0&3&1
\end{tabular}
\end{ruledtabular}
\label{tab:standard-ranks}
\end{localtable}
The direct polynomial definitions confirm these degeneracies. Complete angular kernels are continuous at every listed point. The Sagnac apparent pole is handled in the free basis $(M_1,M_2,S)$ or through analytic numerator derivatives; true filtering zeros retain their zero response.

\subsection{Relation to AET and noise}
Let $\mathsf O_o$ be the orthonormal version of the AET transformation above, with third row $(1,1,1)/\sqrt3$. The map from $A,E,T_o$ to a family is $\mathsf R=\mathsf C\mathsf O_o^{\mathsf T}$. A stable Sagnac form is
\begin{equation}
 \mathsf R_\alpha=\mathsf O_o^{\mathsf T}
 \operatorname{diag}\left(1-z,1-z,-\frac{1+z+z^2}{z-1}\right).
\end{equation}
For any common noise covariance $\mathsf N_o$, use $\mathsf N_{\rm family}=\mathsf R\mathsf N_o\mathsf R^\dagger$, including cross correlations. Retaining two readouts requires selecting the corresponding signal and noise submatrices. Assigning independent unit diagonal noise separately to each family would not compare the same measurement model.

\section{Independent multiprecision quadrature}
\label{app:numerics}
The monopole reference integral can be evaluated without a polarization-frame construction at every node. Let $\boldsymbol u,\boldsymbol v$ be the two unit arms at one vertex, $\boldsymbol u\cdot\boldsymbol v=1/2$, and set $\mu=\bn\cdot\boldsymbol u$, $\nu=\bn\cdot\boldsymbol v$. Combining each incoming link with the delayed outgoing link gives the normalized round-trip factor
\begin{align}
 H_a(\mu)={}&\ii e^{-\ii a(1+\mu)/2}\nonumber\\
 &\times\left[\sinc\frac{a(1+\mu)}2+e^{-\ii a}\sinc\frac{a(1-\mu)}2\right].
\end{align}
With the tensor normalization of Sec.~II, the directional intensity of their difference is
\begin{align}
 \mathcal K_T={}&A(\mu)|H_a(\mu)|^2+A(\nu)|H_a(\nu)|^2\nonumber\\
 &-2B(\mu,\nu)\operatorname{Re}[H_a(\mu)H_a(\nu)^*],\\
 A(\mu)={}&\tfrac14(1-\mu^2)^2,\\
 B(\mu,\nu)={}&\tfrac14\left[2(\tfrac12-\mu\nu)^2-(1-\mu^2)(1-\nu^2)\right].
\end{align}
Thus $g_{00}^T/a^2=(8\pi)^{-1}\int\mathcal K_T\,\dd\Omega$. This integrand follows directly from the complete light-link responses; no integrated coefficients enter it. It is even under reflection in the detector plane and under $\bn\mapsto-\bn$. Using these symmetries, an $N\times4N$ full-sphere Gauss--Legendre--periodic rule requires only $N^2$ evaluated nodes. The reference implementation retains these numerical optimizations and uses multiprecision nodes, geometry, trigonometric functions, and summation throughout. Breathing, longitudinal, and vector checks replace only the invariant polarization contractions.

The grid search uses $N=8,16,24,32,48,64,96,128,192,256,384$ and working precision $p+25$. Each accepted result passes both Eq.~\eqref{eq:mp-tolerance} against the original finite formula at $p+100$ working digits and the successive-grid-change test. The latter supplies an operational stopping indicator; the former supplies an additional validation criterion and is excluded from timing. Neither criterion is an interval certificate. First-call timing includes all attempted grids and node preparation; repeated timing uses the accepted grid. First-call measurements are single runs, whereas repeated timings use the medians stated in Table~\ref{tab:mp-efficiency}. Detailed errors, all attempted grids, software versions, and preparation conventions are retained in the reproducibility records.

The distributed frequency Taylor table has degree 64 and double-precision coefficients; it is used only for ordinary-accuracy evaluation. Higher-precision requests use exact-source multiprecision expressions, with precision refinement. At a nonzero Sagnac delay point, the local expansion in Eq.~\eqref{eq:local-sagnac} instead uses arbitrary-precision Taylor coefficients. The tested local plans retain degree 56, compare with degree 28, and remove the exact constant and linear zeros before dividing. Their preparation cost is paid once per polarization, delay point, degree, and working precision. The intended input offset $\delta$ is supplied directly, avoiding loss of the offset when $2\pi n+\delta$ is first rounded to machine precision.

\bibliographystyle{apsrev4-2}
\bibliography{references}

\end{document}